\documentclass[
  final, 3p, 
  number, 
  sort&compress,
]{elsarticle}
\pdfoutput=1

\usepackage[utf8]{luainputenc}
\usepackage[USenglish]{babel}
\usepackage{csquotes}

\usepackage{amsmath}
\numberwithin{equation}{section}
\allowdisplaybreaks
\usepackage{amssymb}
\usepackage{commath}
\usepackage{mathtools}
\usepackage{bbm}
\usepackage{nicefrac}
\usepackage{adjustbox}

\usepackage{siunitx}
\usepackage{comment}
\usepackage{amsthm}
\usepackage{thmtools}
\usepackage{etoolbox}
\makeatletter
\patchcmd{\thmt@setheadstyle}
 {\bgroup\thmt@space}
 {\thmt@space}
 {}{}
\patchcmd{\thmt@setheadstyle}
 {\egroup\fi}
 {\fi}
 {}{}
\makeatother
\declaretheoremstyle[
  bodyfont=\normalfont\itshape,
  headformat=\NAME\ \NUMBER\NOTE,
]{myplain}
\declaretheoremstyle[
  headformat=\NAME\ \NUMBER\NOTE,
]{mydefinition}
\newcommand{\envqed}{{\lower-0.3ex\hbox{$\triangleleft$}}}
\declaretheorem[style=myplain,numberwithin=section]{theorem}

\declaretheorem[style=mydefinition,numberlike=theorem,qed=\envqed]{remark}

\usepackage[plainpages=false,pdfpagelabels,hidelinks,unicode]{hyperref}

\usepackage{color}
\usepackage{graphicx}
\usepackage[small]{caption}
\usepackage{subcaption}
\usepackage{overpic}

\usepackage{relsize}
\usepackage{adjustbox}
\usepackage{algorithm}
\usepackage[noend]{algpseudocode}
\usepackage{booktabs}
\usepackage{tikz}
\usetikzlibrary{patterns, decorations.pathreplacing, automata, arrows, positioning, calc, patterns.meta}
\usepackage{pgfplots}
\usetikzlibrary{plotmarks}  

\usepackage{cleveref}
\crefname{section}{section}{sections}
\crefname{subsection}{subsection}{subsections}
\crefname{appendix}{}{}
\Crefname{section}{Section}{Sections}
\Crefname{subsection}{Subsection}{Subsections}
\crefname{figure}{Figure}{Figures}
\crefname{definition}{Definition}{Definitions}
\crefname{theorem}{Theorem}{Theorems}
\crefname{lemma}{Lemma}{Lemmas}
\crefformat{equation}{\textup{#2(#1)#3}}
\crefrangeformat{equation}{\textup{#3(#1)#4--#5(#2)#6}}
\crefmultiformat{equation}{\textup{#2(#1)#3}}{ and \textup{#2(#1)#3}}
{, \textup{#2(#1)#3}}{, and \textup{#2(#1)#3}}
\crefrangemultiformat{equation}{\textup{#3(#1)#4--#5(#2)#6}}%
{ and \textup{#3(#1)#4--#5(#2)#6}}{, \textup{#3(#1)#4--#5(#2)#6}}{, and \textup{#3(#1)#4--#5(#2)#6}}
\Crefformat{equation}{#2Equation~\textup{(#1)}#3}
\Crefrangeformat{equation}{Equations~\textup{#3(#1)#4--#5(#2)#6}}
\Crefmultiformat{equation}{Equations~\textup{#2(#1)#3}}{ and \textup{#2(#1)#3}}
{, \textup{#2(#1)#3}}{, and \textup{#2(#1)#3}}
\Crefrangemultiformat{equation}{Equations~\textup{#3(#1)#4--#5(#2)#6}}%
{ and \textup{#3(#1)#4--#5(#2)#6}}{, \textup{#3(#1)#4--#5(#2)#6}}{, and \textup{#3(#1)#4--#5(#2)#6}}
\crefdefaultlabelformat{#2\textup{#1}#3}

\begingroup\expandafter\expandafter\expandafter\endgroup
\expandafter\ifx\csname pdfsuppresswarningpagegroup\endcsname\relax
\else
\fi

\usepackage{booktabs}
\usepackage{rotating}
\usepackage{multirow}

\usepackage{lineno}
\newcommand*\patchAmsMathEnvironmentForLineno[1]{%
  \expandafter\let\csname old#1\expandafter\endcsname\csname #1\endcsname
  \expandafter\let\csname oldend#1\expandafter\endcsname\csname end#1\endcsname
  \renewenvironment{#1}%
     {\linenomath\csname old#1\endcsname}%
     {\csname oldend#1\endcsname\endlinenomath}}%
\newcommand*\patchBothAmsMathEnvironmentsForLineno[1]{%
  \patchAmsMathEnvironmentForLineno{#1}%
  \patchAmsMathEnvironmentForLineno{#1*}}%
\AtBeginDocument{%
\patchBothAmsMathEnvironmentsForLineno{equation}%
\patchBothAmsMathEnvironmentsForLineno{align}%
\patchBothAmsMathEnvironmentsForLineno{flalign}%
\patchBothAmsMathEnvironmentsForLineno{alignat}%
\patchBothAmsMathEnvironmentsForLineno{gather}%
\patchBothAmsMathEnvironmentsForLineno{multline}%
}

\usepackage{xspace}

\definecolor{bluemakie}{rgb}{0.0, 0.44705883,0.69803923}
\definecolor{yellowmakie}{rgb}{0.9019608, 0.62352943, 0.0}
\definecolor{greenmakie}{rgb}{0.0, 0.61960787, 0.4509804}
\definecolor{skybluemakie}{rgb}{0.3372549, 0.7058824, 0.9137255}
\definecolor{orangemakie}{rgb}{0.8352941, 0.36862746, 0.0}

\definecolor{britishracinggreen}{rgb}{0.0, 0.4, 0.12}

\newcommand{\kr}{Knothe-Rosenblatt rearrangement\xspace}

\newcommand{\linad}{linear advection\xspace}

\newcommand{\linpam}{Lin-PAM}
\newcommand{\lo}{Lorenz-63 model\xspace}
\newcommand{\elo}{embedded \lo}

\newcommand{\consSMF}{\text{ConsSMF}\xspace}

\newcommand{\eg}{e.g.,\xspace}
\newcommand{\ie}{i.e.,\xspace}

\newcommand{\BB}{\boldsymbol}
\newcommand{\be}{\begin{equation}}
\newcommand{\ee}{\end{equation}}

\newcommand{\ba}{\begin{align}}
\newcommand{\ea}{\end{align}}

\newcommand{\real}[1]{\mathbb{R}^{#1}}
\newcommand{\complex}[1]{\mathbb{C}^{#1}}

\newcommand{\code}[1]{\texttt{#1}}

\newcommand{\id}[1]{\BB{I}_{#1}}
\newcommand{\zero}[1]{\BB{0}_{#1}}
\newcommand{\one}{\BB{1}}

\newcommand{\indep}{\perp \!\!\! \perp}

\newcommand{\invar}{\textbf{H}}
\newcommand{\C}{\BB{c}} 

\newcommand{\supp}[1]{\operatorname{supp}(#1)}

\newcommand{\A}{\BB{A}}

\newcommand{\U}{\BB{U}}

\newcommand{\T}{\BB{T}}

\newcommand{\Q}{\BB{Q}}
\renewcommand{\L}{\BB{L}}

\newcommand{\eigen}{\BB{\Lambda}}

\newcommand{\diag}[1]{\BB{\mathrm{D}}_{#1}}

\newcommand{\x}{\BB{x}}

\newcommand{\y}{\BB{y}}
\newcommand{\ystar}{\BB{y}^\star}
\newcommand{\X}{\BB{\mathsf{X}}}
\newcommand{\Y}{\BB{\mathsf{Y}}}
\newcommand{\Z}{\BB{\mathsf{Z}}}
\newcommand{\z}{\BB{z}}

\newcommand{\rmse}{\operatorname{RMSE}}
\newcommand{\spread}{\operatorname{spread}}

\newcommand{\stack}[2]{\begin{bmatrix} #1 \\ #2 \end{bmatrix}}

\newcommand{\YX}{\stack{\Y}{\X}}

\newcommand{\Xup}{\boldsymbol{\mathcal{X}}}
\newcommand{\Yup}{\boldsymbol{\mathcal{Y}}}

\newcommand{\Xperpup}{\boldsymbol{\mathcal{X}}_{\perp}}
\newcommand{\Xparaup}{\boldsymbol{\mathcal{X}}_{\parallel}}

\newcommand{\iup}[1]{{#1}^{(i)}}

\newcommand{\dyn}{\BB{f}}
\newcommand{\obs}{\BB{g}}
\newcommand{\Obs}{\BB{G}}

\newcommand{\noiseobs}{\BB{\epsilon}}
\newcommand{\Noisedyn}{\BB{\mathsf{W}}}
\newcommand{\Noiseobs}{\BB{\mathcal{E}}}

\newcommand\given[1][]{\:#1\vert\:}

\newcommand{\enkf}{EnKF\xspace}
\newcommand{\uncons}{unconstrained\xspace}
\newcommand{\cons}{constrained\xspace}

\newcommand{\unenkf}{UnEnKF\xspace}
\newcommand{\consenkf}{ConsEnKF\xspace}

\newcommand{\smf}{SMF\xspace}
\newcommand{\unsmf}{UnSMF\xspace}
\newcommand{\consmf}{ConsSMF\xspace}

\newcommand{\E}[2]{\mathrm{E}_{#1}\left[#2\right]}
\newcommand{\pdf}[1]{\pi_{#1}}

\newcommand{\N}[2]{\mathcal{N}\left(#1, #2\right)}

\newcommand{\cov}[1]{\BB{\Sigma}_{#1}}
\newcommand{\scov}[1]{\widehat{\BB{\Sigma}}_{#1}}
\newcommand{\mean}[1]{\BB{\mu}_{#1}}
\newcommand{\smean}[1]{\widehat{\BB{\mu}}_{#1}}

\newcommand{\meanx}{\mean{\X}}
\newcommand{\covx}{\cov{\X}}

\newcommand{\meany}{\mean{\Y}}
\newcommand{\covy}{\cov{\Y}}

\newcommand{\meanz}{\mean{\Z}}
\newcommand{\covz}{\cov{\Z}}

\newcommand{\Uperp}{\U_\perp}
\newcommand{\Upara}{\U_\parallel}

\newcommand{\xperp}{\x_\perp}
\newcommand{\xpara}{\x_\parallel}
\newcommand{\Xperp}{\X_\perp}
\newcommand{\Xpara}{\X_\parallel}

\newcommand{\push}[1]{{#1}_{\sharp}}

\newcommand{\tmap}{\BB{T}_{\ystar}}
\newcommand{\constmap}{\widetilde{\BB{T}}_{\ystar}}

\newcommand{\tmapkf}{\BB{T}_{\text{KF}}}

\newcommand{\smap}{\BB{S}}
\newcommand{\smappush}{\BB{S}_\sharp}

\newcommand{\lmap}{\BB{L}}

\newcommand{\ddt}[1]{\frac{\partial #1}{\partial t}}

\renewcommand{\div}{\nabla\cdot}

\newcommand{\svec}{\BB{s}}

\newcommand{\domain}{\Omega}

\newcommand{\dtobs}{\Delta t_{\text{obs}}}

\renewcommand{\d}[1]{\mathrm{d}#1}
\newcommand{\dt}{\d{t}}

\renewcommand{\Finv}{\boldsymbol{\mathcal{F}}^{-1}}

\definecolor{blue1}{RGB}{0,76,153}
\definecolor{green1}{RGB}{1,109,57}
\definecolor{red1}{RGB}{153,0,0}

\definecolor{MainBlue}{RGB}{20,80,200}
\definecolor{LightBlue}{RGB}{102,178,255}
\definecolor{MainGreen}{RGB}{34,139,102}
\definecolor{MainOrange}{RGB}{217,95,2}
\definecolor{AccentOrange}{RGB}{230,159,0}
\definecolor{AxisGray}{RGB}{120,120,120}

\newcommand{\Span}{\mathrm{span}}

\renewcommand{\d}{\mathrm{d}}

\newcommand{\orcid}[1]{ORCID:~\href{https://orcid.org/#1}{#1}}

\begin{document}

\begin{frontmatter}

\title{Preserving linear invariants in ensemble filtering methods}

\author[1]{Mathieu Le Provost\fnref{orcidML}}
\fntext[orcidML]{\orcid{0000-0003-0396-5740}}
\address[1]{Department of Aeronautics and Astronautics, Massachusetts Institute of Technology, Cambridge, MA, 02139, USA}

\author[2]{Jan Glaubitz\fnref{orcidJG}\corref{cor1}}
\fntext[orcidJG]{\orcid{0000-0002-3434-5563}}
\ead{jan.glaubitz@liu.se}
\cortext[cor1]{Corresponding author}
\address[2]{Department of Mathematics, Link\"oping University, SE-58183 Linköping, Sweden}

\author[1]{Youssef Marzouk\fnref{orcidYM}}
\fntext[orcidYM]{\orcid{0000-0001-8242-3290}}

\begin{abstract}
  Data assimilation combines dynamical models with observations to improve state estimates. 
Ensemble filters sequentially assimilate observations by updating a set of samples over time, alternating between a forecast and an analysis step.
Accurate and robust predictions often require preserving critical invariants such as mass, stoichiometric balance of chemical species, and electrical charge.
While modern numerical solvers maintain these invariants, existing invariant-preserving analysis steps are limited to Gaussian settings. 
Furthermore, they can be incompatible with regularization techniques such as inflation and covariance tapering.
In this work, we focus on preserving linear invariants in non-Gaussian filtering problems. 
Leveraging tools from measure transport theory, we introduce a novel class of nonlinear ensemble filters that preserve any desired linear invariants. 
Notably, we recover a \cons formulation of the Kalman filter for the special case of the Gaussian setting. 
We also demonstrate how to combine preserving invariants with regularization techniques in the ensemble Kalman filter. 
Numerical experiments illustrate the benefits of preserving linear invariants in both ensemble Kalman filters and transport-based nonlinear ensemble filters.
\end{abstract}

\begin{keyword}
  Data assimilation \sep 
  nonlinear filtering \sep 
  ensemble Kalman filter \sep 
  linear invariants
  
  \MSC
  	65C35 \sep 
	62M20 \sep 
    65N75 
    
    \textit{Reproducibility code:} 
    \href{https://github.com/mleprovost/Paper-Linear-Invariants-Ensemble-Filters}{\textcolor{magenta}{https://github.com/mleprovost/Paper-Linear-Invariants-Ensemble-Filters}} 
\end{keyword}

\end{frontmatter}

\section{Introduction}

Filtering is a potent data assimilation (DA) framework for enhancing the accuracy of numerical simulations of complex physical systems by sequentially incorporating observational data \cite{asch2016data,carrassi2018data,evensen2022data,sanz2023inverse}. 
Consider a state process $\{\X_t\}_{t \geq 0} \in \real{n}$ and observation variables $\Y_t \in \real{d}$, where $t$ represents time.
In the filtering problem, we seek to characterize the so-called filtering distribution $\pdf{\X_t \given \Y_{1:t} = \ystar_{1:t}}$, where $ \ystar_1,\ldots, \ystar_t$ are the realizations of the observation variables up to time $t$. 
The filtering distribution is generally unavailable in closed form and computationally challenging to characterize. 
Ensemble filters are an important class of algorithms for building Monte Carlo approximations of the filtering distribution, by updating a set of samples over time \citep{asch2016data, carrassi2018data}. 
They usually operate in two steps: the forecast step, which propagates each sample through the dynamical model, and the analysis step, which updates the samples based on new observations. 
Notably, the analysis step does not involve time propagation and can be treated as a static inference problem.

It has long been recognized that many dynamical systems conserve specific quantities, such as mass, momentum, Hamiltonians, energy, stoichiometric balance of chemical species, and electrical charge. 
We refer to these conserved quantities as ``invariants.'' 
Specifically, we say that $\invar \colon \real{n} \to \real{r}$ is an invariant for the state process $\{\X_t\}_{t \geq 0}$ if the value of $\invar$ is conserved over time. 
That is, for any $\X_0 \sim \pdf{\X_0}$, we have $\invar(\X_t) = \invar(\X_0)$ for all $t \geq 0$, where $\pdf{\X_0}$ is a distribution for the initial condition. 

To produce physically admissible solutions, modern numerical solvers ensure that discrete solutions mimic the critical invariants of the original system \citep{hairer2006geometric}.
However, DA methods have no intrinsic knowledge of the invariants of the underlying dynamical system and can produce non-physical state estimates. 
For instance, they may result in flow fields with mass imbalance, negative chemical concentrations, or nonzero divergence in incompressible fluid mechanics \citep{albers2019ensemble, janjic2014conservation}. 
We refer readers to \citep{janjic2014conservation,dubinkina2018relevance} for further motivations on preserving invariants in data assimilation schemes. 
Thus, it is critical to incorporate our long-standing physical knowledge into these data assimilation algorithms to produce physically admissible state estimates.

\subsection*{Our contribution}

We propose a new class of filtering methods designed to preserve linear invariants of the form $\invar(\X) = \Uperp^\top \X$, where $\Uperp \in \real{r \times n}$. 
Such invariants are ubiquitous in scientific and engineering applications. Examples include mass conservation, divergence-free constraints in incompressible fluid dynamics \citep{kajishima2016computational}, force balance in statics \citep{craig2006fundamentals}, conservation of electric charge and current in Kirchhoff's laws, and stoichiometric constraints in chemical reaction networks.

To this end, we adopt the perspective of recent works such as \citep{spantini2022coupling, leprovost2021low, leprovost2023adaptive, ramgraber2023_smoothing_part1}, which interpret the analysis step of a filter as a transformation $\BB{T}_{\ystar_t}$---referred to as the \emph{analysis map}---that maps samples $\{ (\iup{\y}, \iup{\x}) \}$ from the joint forecast distribution $\pdf{(\Y_t, \X_t) \given \Y_{1:t-1} = \ystar_{1:t-1}}$ to samples $\{\iup{\x}_a \}$ from the filtering distribution $\pdf{\X_t \given \Y_{1:t} = \ystar_{1:t}}$ via $\iup{\x}_a = \BB{T}_{\ystar_t}(\iup{\y}, \iup{\x})$.
For instance, the ensemble Kalman filter (\enkf) \citep{evensen1994sequential} estimates an \emph{affine} analysis map by replacing the covariance matrices of the classical Kalman update with empirical covariances computed from the joint samples $\{ (\iup{\x}, \iup{\y}) \} \sim \pdf{(\X_t, \Y_t) \given \Y_{1:t-1} = \ystar_{1:t-1}}$; see \citep{spantini2022coupling, leprovost2022lowenkf}.

Building on this formulation, we introduce a new class of (potentially nonlinear) analysis maps, denoted by $\widetilde{\BB{T}}_{\ystar_t}$, that are explicitly designed to \emph{preserve linear invariants}. 
We refer to these as \emph{Linear-invariant-Preserving Analysis Maps (\linpam{}s)}. 
Specifically, we consider any set of $r$ linear invariants of the state variable $\X$ expressed as $\invar(\x) = \Uperp^\top \x$ for some $\Uperp \in \real{r \times n}$.
For any joint forecast sample $(\iup{\y}, \iup{\x})$ with invariant value $\invar(\iup{\x}) = \iup{\C} \in \real{r}$, the proposed \linpam{} satisfies $\invar(\widetilde{\BB{T}}_{\ystar_t}(\iup{\y}, \iup{\x})) = \iup{\C}$. 

To construct such \linpam{}s, we perform a change of variables from the original observation-state space to a rotated coordinate system in which the first $r$ state components represent the linear invariants. 
By carrying out the analysis step in this rotated space while omitting updates to the invariant coordinates, and subsequently lifting the result back to the original space, we ensure exact preservation of invariants. 
Crucially, we show that any empirical estimator $\widehat{\BB{T}}{\ystar_t}$ of a \linpam\ $\widetilde{\BB{T}}{\ystar_t}$---obtained from forecast samples---automatically preserves linear invariants, regardless of the estimator’s accuracy.

In the Gaussian setting, our framework recovers a projected version of the Kalman filter that explicitly preserves linear invariants, as proposed in \citep{amor2018constrained}. 
Moreover, we leverage this construction to develop a novel \emph{invariant-preserving ensemble Kalman filter} that remains compatible with commonly used regularization techniques. 
This resolves an open issue in ensemble Kalman filtering: popular regularization techniques---such as covariance inflation, localization, and tapering \citep{janjic2014conservation, asch2016data, li2024structurally}---often break conservation of invariants.

We demonstrate the performance of the proposed \linpam{}s in a series of computational experiments. 
In a synthetic example with an arbitrary number $r$ of linear invariants, we show that preserving invariants significantly improves performance, especially when the ensemble size $M$ is small and the ratio $r/n$ of invariants to state dimension is large. 
We also assess the benefits of invariant preservation in both linear and nonlinear filtering problems, using a linear advection equation and a low-dimensional chaotic system, respectively. In both cases, we observe that enforcing linear invariants reduces the filter's tracking error.

\subsection*{Related works} 

The problem of enforcing constraints has been studied for the Kalman filter and the \enkf in \citep{simon2010kalman, amor2018constrained, albers2019ensemble, gupta2007kalman, wu2019adding, zhang2020regularized}. Different strategies have been pursued: \citep{simon2010kalman, albers2019ensemble} leveraged a variational formulation of the Kalman filter and the \enkf to enforce hard constraints. 
\citep{gupta2007kalman} enforced hard constraints in the Kalman filter by augmenting the observations with a noiseless version of the constraints. 
\citep{prakash2010constrained} extended this observation augmentation approach to account for soft constraints in the \enkf. 
\citep{janjic2014conservation} used a constrained optimization to enforce linear invariants and positivity constraints in the  \enkf. 
\citep{wu2019adding} enforced soft constraints in the \enkf by a reweighting of the ensemble members. 
\citep{zhang2020regularized} introduced the regularized \enkf to account for various sources of prior knowledge, \eg hard/soft constraints and sparsity.
Furthermore, after the first version of the present work was published on arXiv (see \cite{provost2024preserving}), \cite{subrahmanya2024preserving} proposed a variational extension of the Fokker--Planck equation---which describes a general particle flow filtering framework---to incorporate nonlinear equality state constraints in particle flow filters. 
Specifically, \cite{subrahmanya2024preserving} presents two algorithmic approaches: (i) VFPSTAB, which inexactly preserves constraints by adding a stabilizing drift term, and (ii) VFPDAE, which preserves constraints by treating the VFP dynamics as a stochastic differential-algebraic equation (SDAE). 
The numerical realization of the VFPDAE approach serves as an ``evolve and project'' method, in which the discretized posterior flow dynamics first evolve the solution and then project it onto the constraint manifold. 
While the VFPDAE approach exactly preserves invariants, it is unclear whether it still targets the correct posterior. 

In contrast, the present work adopts a functional perspective and introduces a generic framework for preserving linear invariants in general---potentially non-Gaussian---settings. 
Our method is exact, making it straightforward to implement while preserving constraints rigorously.

\subsection*{Outline}

\Cref{sec:nomenclature} introduces the notation used throughout the paper.
\Cref{sec:filtering} provides an overview of the Bayesian filtering problem and common ensemble-based approaches.
\Cref{sec:transport} introduces key tools from measure transport theory and uses them to derive a general expression for the analysis map.
In \Cref{sec:linPAMs}, we construct novel analysis maps that preserve linear invariants, referred to as \linpam{}s.
\Cref{sec:gaussian} specializes the framework to the Gaussian setting.
\Cref{sec:invariants_kalman} analyzes when the analysis step of the classical Kalman filter preserves linear invariants.
Numerical experiments comparing the \uncons{} and \cons{} variants of the ensemble Kalman filter and the stochastic map filter are presented in \Cref{sec:examples}.
\Cref{sec:conclusion} concludes the paper with a summary of findings and directions for future work. 
\section{Nomenclature \label{sec:nomenclature}}

In the rest of this manuscript, we use the following conventions. 
Serif fonts refer to random variables, e.g., $\BB{\mathsf{Q}}$ on $\real{n}$ or $\mathsf{Q}$ on $\real{}$. 
Lowercase Roman fonts denote realizations of random variables, e.g., $\BB{q}$ on $\real{n}$ or $q$ on $\real{}$. 
$\pdf{\BB{\mathsf{Q}}}$ denotes the probability density function for the random variable $\BB{\mathsf{Q}}$, and $\BB{\mathsf{Q}} \sim \eta$ means that the random variable $\BB{\mathsf{Q}}$ is distributed according to $\eta$. 
If not stated otherwise, we assume that the probability densities have full support.
The mean and covariance matrix of the random variable $\BB{\mathsf{Q}}$ are denoted by $\mean{\BB{\mathsf{Q}}}$ and $\cov{\BB{\mathsf{Q}}}$, respectively. 
The cross-covariance matrix of the random variables $\BB{\mathsf{Q}}$ and $\BB{\mathsf{R}}$ is denoted by $\cov{\BB{\mathsf{Q}}, \BB{\mathsf{R}}}$.
Empirical quantities are differentiated from their asymptotic counterparts using carets above each symbol, e.g., $\smean{\BB{\mathsf{Q}}}$. 
A matrix $\U \in \real{r \times n}$ is a sub-unitary matrix if its columns are orthonormal, i.e., $\U^\top \U = \id{n}$. 
A sub-unitary matrix is orthonormal if $r = n$ and $\U \U^\top = \id{n}$.
\section{Background on the filtering problem}
\label{sec:filtering}

Consider a generic state-space model given by the pair of a dynamical model and an observation model for the state process  $\{ \X_t \}_{t \geq 0}$ and the observation process $\{ \Y_t \}_{t > 0}$. 
The state process $\{ \X_t \}_{t \geq 0}$ is fully described by an initial distribution $\pdf{\X_0}$ and a dynamical model that propagates the state forward in time:
\begin{equation}
\label{eqn:dyn}
    \X_{t+1} = \dyn(\X_t) + \Noisedyn_t, \; \text{for }t \geq 0,
\end{equation}
where $\dyn: \real{n} \to \real{n}$ is the forward operator and the process noise $\Noisedyn_t \in \real{n}$ is independent of the state $\X_t$. 
We do not typically have access to full observations. 
Instead, the state process $\{\X_t \}_{t \geq}$ is only observed through an indirect and perturbed process $\{ \Y_t \}_{t > 0}$, where $\Y_t$ is given by the observation model
\begin{equation}
\label{eqn:obs}
    \Y_{t} = \obs(\X_t) + \Noiseobs_t, \; \text{for } t > 0.
\end{equation}
Here, $\obs \colon \real{n} \to \real{d}$ is the observation operator, and the observation noise variable $\Noiseobs_t$ is independent of $\X_t$. 

The filtering problem is to characterize the filtering density $\pdf{\X_t \given \Y_{1:t} = \ystar_{1:t}}$, which describes the probability of a particular state realization at time $t$ given all the realizations of the observation variable $\Y$ up to that time \citep{asch2016data}. 
The filtering density $\pdf{\X_t \given \Y_{1:t} = \ystar_{1:t}}$ cannot be computed in closed form for generic state-space models and non-Gaussian initial state distributions. 
This limitation motivated the development of empirical approximations of the filtering density. 

It is important to note that through the form of the dynamical and observation models, we are making particular assumptions on the conditional independence structure of the joint distribution of the state and observation processes $\pdf{\X_{0:t}, \Y_{1:t}}$. 
Indeed, we assume that the state process follows a Markov chain, such that the state at time $t$ is conditionally independent of the state realizations at previous times given the state at time $t-1$, i.e., $\X_{t} \indep \X_{s} \given \X_{t-1}$ for any $s \leq t-1$. 
This implies $\pdf{\X_{t} \given \X_{1:t-1}} = \pdf{\X_{t} \given \X_{t-1}}$. 
Similarly, we assume that the observation variable $\Y_t$ at time $t$ is conditionally independent of the state realizations at previous times given the state at time $t$, i.e., $\pdf{\Y_t \given \X_{1:t}} = \pdf{\Y_t \given \X_t}$. 
Under these conditional independence assumptions, we can factorize the joint distribution $\pdf{\X_{0:T}, \Y_{1:T}}$ as
\begin{equation}\label{eq:joint_distr}
    \pdf{\X_{0:T}, \Y_{1:T}} 
        = \underbrace{\pdf{\X_0}}_{\text{initial state}} 
        \prod_{t=1}^T \underbrace{\pdf{\Y_t \given \X_t}}_{\text{observations}} \underbrace{\pdf{\X_t \given \X_{t-1}}}_{\text{dynamics}}.
\end{equation}
From \cref{eq:joint_distr}, we derive a recursive relation to propagate the filtering distribution from time $t$ to $t-1$:
\begin{equation}
\label{eqn:filtering_update}
    \pdf{\X_t \given \Y_{1:t}} 
    		\propto \pdf{\Y_t \given \X_t} \, \pdf{\X_{t} \given \Y_{1:t-1}} 
		= \pdf{\Y_t \given \X_t} \int \pdf{\X_t \given \X_{t-1}} \, \pdf{\X_{t-1} \given \Y_{1:t-1}} \, \d \X_{t-1}
\end{equation}
This recursive update operates in two steps: 
First, the forecast distribution $\pdf{\X_t \given \Y_{1:t-1}}$ is obtained by propagating the filtering distribution $\pdf{\X_{t-1} \given \Y_{1:t-1}}$ through the transition kernel $\pdf{\X_t \given \X_{t-1}}$ of the dynamical model \eqref{eqn:dyn}, i.e., $\pdf{\X_t \given \Y_{1:t-1}} = \int \pdf{\X_t \given \X_{t-1}} \pdf{\X_{t-1} \given \Y_{1:t-1}} \d\X_{t-1}$. 
This last equation is known as the Chapman-Kolmogorov equation \citep{asch2016data, carrassi2018data}. Second, we apply Bayes' rule to condition the forecast distribution on the realization of the observation at time $t$, resulting in the filtering distribution at time $t$.

Solving the filtering problem analytically is generally infeasible outside the linear-Gaussian setting \citep{asch2016data}. 
We thus generally rely on Monte Carlo approximations of the filtering density $\pdf{\X_t \given \Y_{1:t}}$ by propagating a set of samples $\{\iup{\x} \}_{i=1}^M$ over time. 
Here, we consider \emph{ensemble filtering methods} \citep{carrassi2018data, spantini2022coupling}, which assign equal weights to the samples. 
These algorithms mimic the two-step recursive update of the filtering distribution in \eqref{eqn:filtering_update} at the sample level. 
First, in the \emph{forecast step}, each sample is propagated through the dynamical model \eqref{eqn:dyn}, yielding a Monte Carlo approximation of the forecast distribution $\pdf{\X_t \given \Y_{1:t-1}}$. 
Second, in the \emph{analysis step}, we condition the forecast samples on the observation variable's realization $\ystar_t$ at time $t$.  
\section{Background on ensemble transport methods \label{sec:transport}}

We briefly review (triangular) transport maps. 
In particular, we recall some of their appealing properties for conditional inference and how they can be leveraged to construct analysis maps \citep{spantini2022coupling}.

\subsection{Motivation for ensemble transport methods}
\label{subsec:transport_motivation}

To condition forecast samples on new observation realizations in the analysis step, we leverage the formalism of measure transport \citep{marzouk2016sampling, spantini2022coupling}. 
To this end, we rely on the existence of a transformation $\BB{T}_{\ystar_t}$, called \emph{prior-to-posterior transformation} or \emph{analysis map}, that transforms the joint forecast distribution $\pdf{(\Y_t, \X_t) \given \Y_{1:t-1} = \ystar_{1:t-1}}$ into the filtering distribution $\pdf{\X_t \given \Y_{1:t} = \ystar_{1:t}}$. 
Under Gaussian assumptions on the joint forecast distribution of the states and observations, the analysis map corresponds to the widely-used Kalman filter update \citep{kalman1960new, spantini2022coupling}. 
In practice, the analysis map must be approximated from samples $\{(\iup{\y}, \iup{\x}) \}$ of the joint forecast distribution $\pdf{(\Y_t, \X_t) \given \Y_{1:t-1} = \ystar_{1:t-1}}$. 
For instance, the ensemble Kalman filter (\enkf) introduced by \cite{evensen1994sequential} relies on a Monte Carlo approximation of the Kalman gain $\cov{\X_t,\Y_t} \cov{\Y_t}^{-1}$ from the joint forecast samples.

Outside the Gaussian setting that underpins the Kalman filter and \textit{t}-distributions \citep{leprovost2023adaptive}, it is challenging to analytically obtain analysis maps for an arbitrary joint forecast distribution $\pdf{(\X_t, \Y_t) \given \Y_{1:t-1} = \ystar_{1:t-1}}$. 
To address these limitations, \cite{spantini2022coupling} used measure transport theory to estimate analysis maps in non-Gaussian settings. 
This active field of research aims at characterizing a target distribution $\pdf{}$ as the transformation of a simpler reference distribution $\eta$ by a map $\smap$ \citep{villani2009optimal, marzouk2016sampling, peyre2019computational}. 
This transport-based methodology provides a principled generalization of the linear Kalman filter to nonlinear analysis maps, producing consistent inference for non-Gaussian filtering problems. 
The resulting ensemble filter is called the \emph{stochastic map filter (SMF)}. 
Transport-based analysis approaches have since been explored further in works such as \cite{chipilski2023exact,al2023optimal,al2024nonlinear}.

\subsection{Overview of triangular transport methods \label{subsec:triangular_transport}}

We next recall some introductory elements on \emph{triangular} transport methods \citep{marzouk2016sampling}.
Given a target distribution with density $\pdf{}: \real{m} \to \real{}$, it is helpful to describe it as the transformation of a simpler reference distribution with density $\eta: \real{m} \to \real{}$ by a map $\smap \colon \real{m} \to \real{m}$. 
A bijective and differentiable map $\smap$ that transforms the distribution $\pdf{}$ into $\eta{}$ is called a transport map, and we say that $\smap$ ``pushes forward'' $\pdf{}$ to $\eta$, denoted $\smappush \pdf{} = \eta$ \citep{marzouk2016sampling}. The formula for the push-forward distribution $\smappush \pdf{}$ corresponds to the classical change of variables in multivariate calculus: 
\begin{equation}
\label{eqn:pushforward}
    \smappush \pdf{}(\z)  = \pdf{}(\smap^{-1}(\z)) \det \nabla_{\z}\smap^{-1}(\z)
\end{equation}
Transport maps are appealing for sampling as i.i.d.\ samples $\{\iup{\z} \}$ from  $\eta$ get mapped to i.i.d. samples $\{\smap(\iup{\z}) \}$ from $\pdf{}$; see \citep{marzouk2016sampling}. 
Building such transformations $\smap$ is the core topic of measure transport theory \citep{marzouk2016sampling}, classically viewed through the perspective of cost minimization \citep{villani2009optimal, peyre2019computational}. 
In this work, we take a different perspective and focus on transformations with appealing properties for Bayesian inference. 
Specifically, we are interested in transport maps tailored for sampling from the conditionals of a joint distribution. 

For the analysis step of the filtering problem, we have samples from the joint forecast distribution $\pdf{(\Y_t, \X_t) \given \y_{1:t-1}^{\star}}$ and seek to generate samples from the filtering distribution $\pdf{\X_t \given \y_{1:t}^{\star}}$. 
Among the transport maps pushing forward $\pdf{}$ to $\eta$, we consider the \kr \citep{rosenblatt1952remarks}. 
It is defined as the unique lower triangular and strictly increasing transformation $\smap \colon \real{m} \to \real{m}$ and is of the form
\begin{equation}
\label{eqn:kr_rearrangement}
     \smap(\z) = \smap(z_1, z_2, \cdots, z_m)=\left[\begin{array}{l}
    S^{1}\left(z_{1}\right) \\
    S^{2}\left(z_{1}, z_{2}\right) \\
    \vdots \\
    S^{m}\left(z_{1}, z_{2}, \ldots, z_{m}\right)
    \end{array}\right],
\end{equation}
where the strict monotonicity of $S^k: \real{k} \to \real{}$ signifies that the univariate function $\xi \mapsto S^k(\z_{1:k-1}, \xi)$ is strictly monotonically increasing for all $z_{1}, z_{2}, \ldots, z_{k-1}$. The lower-triangular structure of the \kr offers both theoretical and computational benefits. If the reference density $\eta$ can be factorized, i.e., $\eta(\z) = \prod_{i=1}^m \eta_i(z_i)$, then \cite{marzouk2016sampling} showed that the univariate function $\xi \mapsto S^k(\z_{1:k-1}, \xi)$ resulting from fixing the first $k-1$ entries pushes forward the conditional distribution $\pdf{\Z_k \given \Z_{1:k=1}}(\xi \given \z_{1:k-1})$ to the $k$th component of the reference density, i.e., $S^k(\z_{1:k-1}, \cdot)_\sharp \pdf{\Z_k \given \Z_{1:k=1}}(\xi \given \z_{1:k-1}) = \eta_k$. 
Once the map $\smap$ has been determined, one can easily represent any conditional of the target distribution. In the next section, we will show how to leverage this property to construct the analysis map. We note that monotone lower-triangular transformations are also computationally attractive, as the determinant of their Jacobian reduces to the product of the partial derivatives of each map component with respect to its last entries. Moreover, the inversion of lower-triangular transformations reduces to a sequence of univariate root-finding problems \citep{marzouk2016sampling}. Remarkably, the \kr is known in the Gaussian case (see Remark \ref{remark:KR_gaussian} below) and has recently been identified for multivariate \textit{t}-distributions by \cite{leprovost2023adaptive}.

\begin{remark}\label{remark:KR_gaussian}
Consider $\X \sim \pdf{\X} = \N{\mean{}}{\cov{}}$ and the Cholesky factorization $\lmap^\top \lmap = \cov{}^{-1}$. 
Then $\smap(\x) = \lmap (\x - \mean{})$ is the \kr  that pushes forward $\pdf{\X}$ to $\eta = \N{\zero{n}}{\id{n}}$. 
The proof uses properties on linear transformations of Gaussian variables, and is omitted for brevity. 
We also note that the above transformation is equivalent to the control-variable transform \cite{parrish1992national,menetrier2015overlooked,evensen2022data} in variational methods, used for preconditioning the variational inverse problem and ensuring that the minimal eigenvalue of the Hessian is 1. 
\end{remark}

\subsection{Construction of the analysis map}
\label{subsec:analysis_map}

We now revisit the construction of the analysis map $\tmap$ that pushes forward the  joint forecast distribution $\pdf{(\Y_t, \X_t) \given \Y_{1:t-1} = \ystar_{1:t-1}}$ to the filtering distribution $\pdf{\X_t \given \Y_{1:t} = \ystar_{1:t}}$. 
Henceforth, we omit the time-dependent subscripts on the variables, since the analysis step involves a static Bayesian inverse problem \citep{spantini2022coupling, leprovost2022lowenkf}. 
The joint random variable $(\Y, \X)$ will refer to the joint forecast random variable $(\Y_{t} , \X_{t}) \given \Y_{1:t-1} = \ystar_{1:t-1}$.

Consider the \kr $\smap$ that pushes forward the joint distribution of the observations and states  $\pdf{\Y, \X}$ to the product reference distribution $\eta= \eta_{\Y} \otimes \eta_{\X}$ with $\eta_{\Y}: \real{d} \to \real{}$ and $\eta_{\X}: \real{n} \to \real{}$. 
From its lower triangular structure, $\smap$ can be partitioned into two blocks: 
\begin{equation}\label{eqn:split}
	\smap(\y, \x)=\left[\begin{array}{c}
	\begin{aligned}
		& \smap^{\Yup}(\y) \\
		& \smap^{\Xup}(\y, \x)
	\end{aligned}
	\end{array}\right],
\end{equation}
where $\smap^{\Yup}\colon\real{d} \xrightarrow{} \real{d}$ and $\smap^{\Xup}\colon\real{d} \times \real{n} \xrightarrow{} \real{n}$. 
\cite[Theorem 2.4]{baptista2020conditional} showed that if $\smap$ is lower triangular and pushes forward $\pdf{\Y, \X}$ to a product reference distribution $\eta= \eta_{\Y} \otimes \eta_{\X}$, then $\smap^{\Xup}$ pushes forward $\pdf{\X \given \Y = \y}$ to $\eta_{\X}$. 
As the distributions $\eta_{\Y}$ and $\eta_{\X}$ are user-specified, one can use standard Gaussian distributions. 
Once the map $\smap^{\Xup}$ has been learned, one can easily generate samples from $\pdf{\X \given \Y = \ystar}$ from $\eta_{\X}$. 

In practice, the learned map $\smap^{\Xup}$ is imperfect, such that errors in the empirical map estimate $\widehat{\smap}^{\Xup}$ will be propagated in the approximation of the conditional distribution via the pullback distribution $\widehat{\smap}^{\Xup}(\ystar, \cdot)^{\sharp} \eta_{\X}$. 
To reduce error propagation in the conditional distribution,  \cite{spantini2022coupling} considered a composite analysis map built by partial inversion of the map $\smap^{\Xup}$. 
As a consequence of the pushforward relation, we have that the conditional map $\x \mapsto \smap^{\Xup}(\y, \x)$ is a bijection on $\real{n}$ for any $\y \in \real{d}$. Let $(\y,\x)$ be a joint sample from $\pdf{\Y, \X}$. For a realization $\ystar$ of the observation variable $\Y$ that we want to condition on, there exists a unique element $\x_{a} \in \real{n}$ such that $\smap^{\Xup}(\ystar, \x_{a}) = \smap^{\Xup}(\y, \x)$. 
This element $\x_a$ is precisely the posterior update of the state $\x$ given $\ystar$; see \citep{spantini2022coupling, leprovost2021low}. 
Finally, this yields the following analysis map $\tmap \colon \real{d} \times \real{n} \to \real{n}$:
\begin{equation}
\label{eqn:analysis_map}
    \tmap(\y, \x) = \smap^{\Xup}(\ystar, \cdot)^{-1} \circ \smap^{\Xup}(\y, \x),
\end{equation}
where the notation $\smap^{\Xup}(\ystar, \cdot)^{-1}$ denotes the inversion of the map $\x \mapsto \smap^{\Xup}(\ystar, \x)$ for fixed $\ystar \in \real{d}$. 
We conclude by connecting the Kalman filter and the SMF to the analysis map of \eqref{eqn:analysis_map} in Remarks \ref{remark:enkf} and \ref{remark:smf}.

\begin{remark}[The Kalman filter]\label{remark:enkf}
Consider two random variables $\X \in \real{n}$ and $\Y \in \real{d}$ such that $(\Y, \X)$ is jointly Gaussian, i.e., 
\begin{equation}
    \stack{\Y}{\X} \sim \N{\stack{\meanx}{\mean{\Y}}}{
    \begin{bmatrix}
    \cov{\Y} & \cov{\X, \Y}^\top \\
    \cov{\X, \Y} &  \covx
    \end{bmatrix}}.
\end{equation}
Then for $\y \in \real{d}$, $\pdf{\X \given \Y = \y} = \N{\mean{\X \given \Y = \y}}{\cov{\X \given \Y = \y}}$ with $\mean{\X \given \Y = \y} = \mean{\X} + \cov{\X, \Y} \cov{\Y}^{-1}(\y - \mean{\Y})$ and $\cov{\X \given \Y = \y} = \cov{\X} - \cov{\X, \Y} \cov{\Y}^{-1} \cov{\X, \Y}^\top$. Let $\L_{\X \given \Y = \y}^\top \L_{\X \given \Y = \y} = \cov{\X \given \Y = \y}^{-1}$ be the Cholesky factorization of $\cov{\X \given \Y = \y}^{-1}$. Then $\smap^{\Xup}(\y, \x) = \L_{\X \given \Y = \y}(\x -  \mean{\X \given \Y = \y})$. By applying \eqref{eqn:analysis_map}, we get $\tmap(\y, \x) = \x - \cov{\X, \Y} \cov{\Y}^{-1}(\y - \ystar)$, recovering the analysis map of the Kalman filter as noted by~\cite{spantini2022coupling}. Thus, \eqref{eqn:analysis_map} corresponds to a generalization of the Kalman filter for non-Gaussian joint distribution of the observations and states $\pdf{\Y, \X}$.
\end{remark}

\begin{remark}[The stochastic map filter (SMF)\label{remark:smf}]
The SMF in \citep{spantini2022coupling} builds an estimator $\widehat{\smap}^{\Xup}$ for $\smap^{\Xup}$ from joint forecast samples of the states and observations. The estimator $\widehat{\smap}^{\Xup}$ is based on a parsimonious expansion of radial basis functions estimated by solving decoupled and convex optimization problems. 
The filtering samples are obtained by applying the estimated analysis map of \eqref{eqn:analysis_map} from $\widehat{\smap}^{\Xup}$ to the joint forecast samples. 
Thus, the SMF can be viewed as a nonlinear generalization of the ensemble Kalman filter.
\end{remark}

\section{Linear invariant-preserving analysis maps (\linpam s)}
\label{sec:linPAMs}

In this section, we propose a new methodology for constructing analysis maps that preserve linear invariants of the dynamics model. 
Consider the linear invariants $\invar(\x) = \Uperp^\top \x$. We will design an analysis map $\tmap$ that pushes forward $\pdf{\Y,\X}$ to $\pdf{\X \given \Y = \ystar}$ while ensuring the invariant $\invar$ is preserved, i.e., $\invar(\tmap(\y, \x)) = \invar(\x)$ for all $(\y, \x) \in \real{d} \times \real{n}$.
We call an analysis map satisfying this property a \emph{linear invariant-preserving analysis map (\linpam)}.
Our approach applies to arbitrary nonlinear analysis maps and non-Gaussian joint distributions of the observations and states, $\pdf{\Y, \X}$.

Assume that the matrix $\Uperp \in \real{n \times r}$ has rank $r$. 
Furthermore, without loss of generality (see Remark \ref{rem:assumption_subUnitary} below), we assume that $\Uperp$ is a sub-unitary matrix, i.e., $\Uperp$ has orthonormal columns. 
We can then uniquely decompose the state $\x \in \real{n}$ as 
\begin{equation}
\label{eqn:state_decomposition}
    \x = \Uperp \Uperp^\top  \x + (\id{} - \Uperp \Uperp^\top) \x = \Uperp \xperp \oplus  \Upara \xpara,
\end{equation}
where $\xperp = \Uperp ^\top \x \in \real{r}$ and $\xpara = \Upara^\top \x \in \real{n -r}$.  
One can use a QR factorization of $\Uperp$ to build $\Upara$.
Thus, we have $\x = \U [\xperp; \xpara]$, or, equivalently, $[\xperp; \xpara] = \U^{-1} \x = \U^\top \x$. 
The decomposition \eqref{eqn:state_decomposition} allows us to reformulate the problem of building an analysis map that preserves linear invariants $\invar(\x) = \Uperp^\top \x$ as building an analysis map that preserves the state components in the span of $\Uperp$. 

\begin{remark}\label{rem:assumption_subUnitary} 
    Consider a linear constraint $\Uperp^\top \x = \C_{\perp} \in \real{r}$, where $\Uperp^\top$ is not necessarily a sub-unitary matrix. 
    We can then use a thin QR factorization $\Uperp = \BB{Q}_{\perp} \BB{R}_{\perp}$, where $\BB{Q}_{\perp} \in \real{n \times r}$ is sub-unitary and $\BB{R}_{\perp} \in \real{r \times r}$ is lower triangular \citep{golub2013matrix}, to rewrite the linear constraint as $\BB{Q}_{\perp}^\top \x = \BB{R}_{\perp}^{-\top} \C_{\perp} \in \real{r}$. 
    Notably, $\BB{Q}_{\perp}^\top$ is a a sub-unitary matrix, i.e., $\BB{Q}_{\perp}^\top$ has orthonormal columns. 
\end{remark}

\subsection{Formulation of the analysis map in the rotated space \label{subsec:analysis_rotatedspace}}

Building upon \eqref{eqn:state_decomposition}, we can apply the change of coordinates $(\Y, \X) \mapsto (\Y, \Xperp, \Xpara)$ with $\Xperp = \Uperp^\top$ and $\Xpara = \Upara^\top$ to express the analysis map of \eqref{eqn:analysis_map} in the rotated space. 
The distribution $\pdf{\Xperp, \Xpara}$ is given by the pushforward of $\pdf{\X}$ by the linear transformation $\U$, i.e., $\pdf{\Xperp, \Xpara} = \U_{\sharp} \pdf{\X}$. Using the push forward formula \eqref{eqn:pushforward}, we thus have 
\begin{equation}
    \pdf{\Xperp, \Xpara}(\xperp, \xpara) 
        = \pdf{\X}(\U [\xperp; \xpara]) \det \nabla (\U [\xperp; \xpara]) 
        = \pdf{\X}(\Uperp \Uperp^\top \x +  \Upara \Upara^\top \x) 
        = \pdf{\X}(\x).
\end{equation}
Hence, we have the following factorization of $\pdf{\Y,\X}$:
\begin{equation}
\label{eqn:factorization_joint}
\begin{aligned}
\pdf{\Y, \X}(\y, \x) = \pdf{\Y, \Xperp, \Xpara}(\y, \xperp, \xpara) = \pdf{\Y}(\y) \pdf{\Xperp \given \Y}(\xperp \given \y) \pdf{\Xpara \given \Y, \Xperp}(\xpara \given \y, \xperp)
\end{aligned}
\end{equation}
Let us now consider the \kr $\smap \colon \real{d} \times \real{r} \times \real{n - r}$ that pushes forward $\pdf{\Y, \Xperp, \Xpara}$ to the product of standard Gaussian references $\eta_{\Y} \otimes \eta_{\Xperp} \otimes \eta_{\Xpara}$, where $\eta_{\Y}, \eta_{\Xperp}, \eta_{\Xperp}$ is defined on $\real{d}, \real{r}, \real{n-r}$, respectively. 
From its lower triangular structure, we can partition $\smap$ as
\begin{equation}
    \smap(\y, \xperp, \xpara) = \left[\begin{array}{c}
\begin{aligned}
        & \smap^{\Yup}(\y) \\
        & \smap^{\Xperpup}(\y, \xperp) \\
        & \smap^{\Xparaup}(\y, \xperp, \xpara)
\end{aligned}
\end{array}\right],
\end{equation}
where $\smap^{\Yup} \colon \real{d} \to \real{d}, \smap^{\Xperpup} \colon \real{d} \times \real{r} \to \real{r}, \text{and } \smap^{\Xparaup} \colon \real{d} \times \real{r} \times \real{n - r} \to \real{n-r}$. 
From the lower block structure of $\smap$ and \cite[Theorem 2.4]{baptista2020conditional}, we have the following relations between the conditionals of $\pdf{\Y, \Xperp, \Xpara}$ and the marginals of $\eta$:
\begin{equation}
    \begin{aligned}
        &\push{\smap^{\Yup}} \pdf{\Y} = \eta_{\Y},\\
        &\push{\smap^{\Xperpup}} \pdf{\Xperp \given \Y} = \eta_{\Xperp},\\
        &\push{\smap^{\Xparaup}} \pdf{\Xpara \given \Y, \Xperp} = \eta_{\Xpara}
    \end{aligned}
    \label{eqn:relation_conditional_marginal}
\end{equation} 
Using the relations in \eqref{eqn:relation_conditional_marginal}, we proceed as follows to sample from the conditional distribution $\pdf{\X \given \Y = \ystar}$, or equivalently $\pdf{\Xperp, \Xpara \given \Y = \ystar}$. Let $(\y, \xperp, \xpara)$ be a joint sample from $\pdf{\Y, \Xperp, \Xpara}$. Following the derivation of the analysis map in Section \ref{subsec:analysis_map}, we generate samples from $\pdf{\Xperp \given \Y = \ystar}$ by seeking the solution $\x_{\perp, a} \in \real{r}$ of 
\begin{equation}
\smap^{\Xperpup}(\ystar, \x_{\perp, a}) = \smap^{\Xperpup}(\y, \xperp).
\end{equation}
Formally, we therefore get the analysis map $\tmap^{\perp}: \real{d} \times \real{r} \to \real{r}$ that pushes forward  $\pdf{\Y, \Xperp}$ to $\pdf{\Xperp \given \Y = \ystar}$ as
\begin{equation}
    \label{eqn:tmap_perp}
    \tmap^{\perp}(\y,  \xperp) = \smap^{\Xperpup}(\ystar, \cdot)^{-1} \circ \smap^{\Xperpup}(\y, \xperp).
\end{equation}
Then we generate samples from $\pdf{\Xpara \given \ystar, \x_{\perp, a}}$ by seeking the solution $\x_{\parallel, a} \in \real{n-r}$ of 
\begin{equation}
\smap^{\Xparaup}(\ystar, \x_{\perp, a}, \x_{\parallel, a}) = \smap^{\Xparaup}(\y, \xperp, \xpara).
\end{equation}
Similar to \eqref{eqn:tmap_perp}, we write the analysis map $\T_{\ystar, \x_{\perp, a}}^{\parallel}: \real{d} \times \real{r} \times \real{n -r }\to \real{n - r}$ that pushes forward $\pdf{\Y, \Xpara, \Xperp}$ to $\pdf{\Xpara \given \Y = \ystar, \Xperp = \x_{\perp, a}}$ as
\begin{equation}
    \label{eqn:tmap_para}
    \T_{\ystar, \x_{\perp, a}}^{\parallel}(\y,  \xperp, \xpara) = \smap^{\Xparaup}(\ystar, \tmap^{\perp}(\y,  \xperp), \cdot)^{-1} \circ \smap^{\Xparaup}(\y, \xperp, \xpara).
\end{equation}
Thus, the analysis map $\tmap \colon \real{d} \times \real{n} \to \real{n}$ that pushes forward $\pdf{\Y, \X}$ to $\pdf{\X \given \Y = \ystar}$ is
\begin{equation}
\label{eqn:analysis_map_block}
    \tmap(\y, \x) = \begin{bmatrix}
        \Uperp, \Upara
    \end{bmatrix} \begin{bmatrix}
        \x_{\perp, a} \\
        \x_{\parallel, a}
    \end{bmatrix} = \Uperp \tmap^{\perp}(\y, \Uperp^\top \x) + \Upara \T_{\ystar, \x_{\perp, a}}^{\parallel}(\y, \Uperp^\top \x, \Upara^\top \x).
\end{equation}
Notably, \eqref{eqn:analysis_map_block} provides an alternative formulation for the analysis map in \eqref{eqn:analysis_map} by building analysis maps $\tmap^{\perp}$ and  $\T^\parallel_{\ystar, \x_{\perp, a}}$ in the spaces spanned by the columns of $\Uperp$ and $\Upara$. We stress that \eqref{eqn:analysis_map_block} and \eqref{eqn:analysis_map} are strictly equivalent as \eqref{eqn:analysis_map_block} relies on the factorization of $\pdf{\Y, \X}$ into $\pdf{\Y} \pdf{\Xperp \given \Y} \pdf{\Xpara \given \Xperp, \Y}$ instead of $\pdf{\Y} \pdf{\X \given \Y}$ for \eqref{eqn:analysis_map}. 
Interestingly, \eqref{eqn:analysis_map_block} proposes to perform the inference by first rotating the state along the columns of $\U$, then performing the inference for $\Xperp$ followed by $\Xpara$, and finally lifting the result to the original space.

\subsection{Preserving the linear invariants in the rotated space \label{subsec:constrained_analysismap}}

The state decomposition \eqref{eqn:state_decomposition} suggests that preserving the linear invariants $\invar(\x) = \Uperp^\top \x$ corresponds to generating posterior samples from $\pdf{\X \given \Y = \ystar, \Xperp = \xperp}$. 
To do so, we modify the analysis map \eqref{eqn:analysis_map_block} by omitting the update of $\Xperp$.  
This is equivalent to constraining the analysis map $\constmap^{\perp}$ to be the identity, i.e., $\constmap^{\perp}(\y, \xperp) = \xperp$. 
In this way, we obtain the \cons analysis map $\constmap^{\parallel}$ as
\begin{equation}\label{eqn:constrained_tmap_para}
    \constmap^{\parallel}(\y, \xperp, \xpara) = \smap^{\Xparaup}(\ystar, \xperp, \cdot)^{-1} \circ \smap^{\Xparaup}(\y, \xperp, \xpara).
\end{equation}
Note that the second argument of ${\smap^{\Xparaup}}^{-1}$ is $\xperp$ in \eqref{eqn:constrained_tmap_para} instead of $\tmap^{\perp}(\y,  \xperp)$ as in \eqref{eqn:tmap_para}. 
Finally, the analysis map $\constmap$ formulated in the original space preserving the invariant $\invar(\x) = \Uperp^\top \x$ reads 
\begin{equation}
\label{eqn:constrained_analysis_map}
\begin{aligned}
        \constmap(\y, \x) & = \Uperp \constmap^{\perp}(\y, \Uperp^\top \x) + \Upara \constmap^{\parallel}(\y, \Uperp^\top \x, \Upara^\top \x)\\
        & =  \Uperp \Uperp^\top \x  + \Upara \constmap^{\parallel}(\y, \Uperp^\top \x, \Upara^\top \x).
\end{aligned}
\end{equation}
The analysis maps \cref{eqn:constrained_tmap_para,eqn:constrained_analysis_map} form the cornerstone of the proposed \linpam \ methodology. 

In practice, an empirical estimator $\widehat{\BB{T}}_{\ystar}^{\parallel}$ for the constrained analysis map $\constmap^{\parallel}$ of \eqref{eqn:constrained_tmap_para} is built at each assimilation cycle from joint forecast samples $\{\iup{\y}, \iup{\x} \}$ of $\pdf{(\Y_t, \X_t) \given \Y_{1:t-1} = \ystar_{1:t-1}}$. Then, we obtain the resulting constrained estimator $\widehat{\BB{T}}_{\ystar}$ for the \linpam\ $\constmap$ of \eqref{eqn:constrained_analysis_map} by replacing  $\constmap^{\parallel}$  with $\widehat{\BB{T}}_{\ystar}^{\parallel}$. Remark \ref{remark:empirical_linpam} shows that an empirical estimator $\widehat{\BB{T}}_{\ystar}$ of a \linpam\ still preserves linear invariants. Remark \ref{remark:constant_prior} discusses the construction of \linpam\ of \eqref{eqn:constrained_analysis_map} when the linear invariants are constant over the prior. 

\begin{algorithm}[h]
\caption{\texttt{consSMF}$(\ystar, \pdf{\Y \mid \X = \cdot}, \{\x^i\})$ assimilates $\ystar$ into $\{\x^i\}_{i=1}^M$ while preserving linear invariants}
\label{algo:nonlinear_constrained}
\begin{algorithmic}[1]
	\State \textbf{Input:} $\ystar \in \mathbb{R}^d$, likelihood model $\pdf{\Y \mid \X = \cdot}$, samples $\{\x^i\}_{i=1}^M$ from $\pdf{\X}$
	\State \textbf{Output:} Samples $\{\x_a^i\}_{i=1}^M$ from $\pdf{\X \mid \Y = \ystar}$
	\State{Generate $M$ likelihood samples $\{ \y^i \}_{i=1}^M$ by drawing from $\pdf{\Y \mid \X = \x^i}$}
	\State{Project forecast samples onto $\Uperp$ and $\Upara$: $\xperp^i = \Uperp^\top \x^i$ and $\xpara^i = \Upara^\top \x^i$}
	\State{Estimate the transport map $\smap^{\Xparaup}$ from samples $\{ (\y^i, \xperp^i, \xpara^i) \}_{i=1}^M$}
	\State{Perform \texttt{cons} analysis in transformed space: $$\x_{\parallel, a}^i = \constmap^\parallel(\y^i, \xperp^i, \xpara^i) = \left[\smap^{\Xparaup}(\ystar, \xperp^i, \cdot)\right]^{-1} \circ \smap^{\Xparaup}(\y^i, \xperp^i, \xpara^i)$$}
	\State{Lift result to original space: $\x_a^i = \Uperp \xperp^i + \Upara \x_{\parallel, a}^i$}
	\State \Return $\{\x_a^i\}_{i=1}^M$
\end{algorithmic}
\end{algorithm}

\begin{figure}[tb]
    \centering
    \resizebox{0.75\textwidth}{!}{%
    \begin{overpic}[width = 0.8\linewidth]{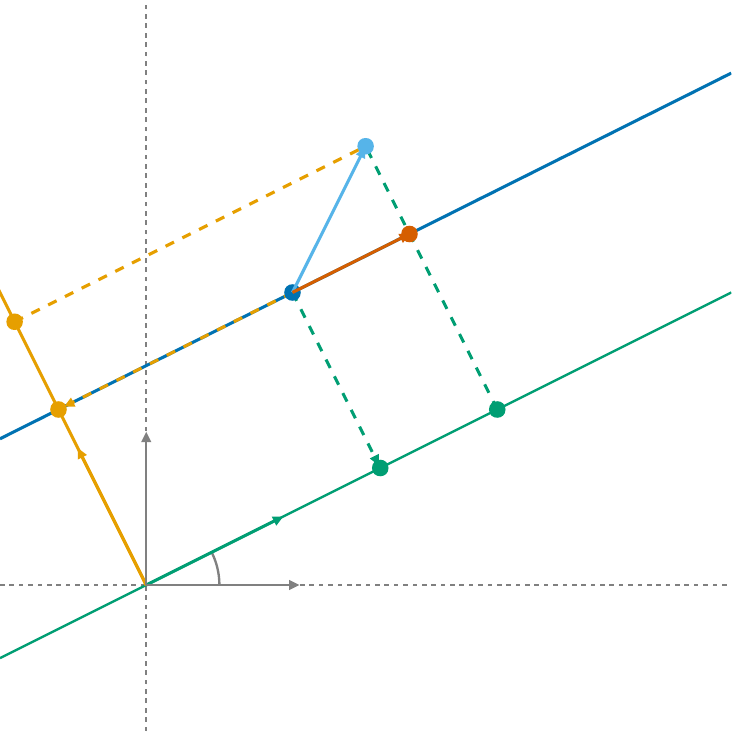}
    
    \put(44,58){$\x$}
    \put(48,82){$\tmap(\y, \x)$}
    \put(54,72){$\constmap(\y, \x)$}

    \put(16,56){Project on $\operatorname{span}(\Uperp)$}
    \put(13,42){$ \Uperp \xperp$}
    \put(4,64){$\Uperp \tmap^{\perp}(\y, \xperp)$}

    \put(48,50){Project on $\operatorname{span}(\Upara)$}
    \put(42,38){$\Upara \xpara$}
    \put(62,38){$\Upara \constmap^{\parallel}(\y, \xperp, \xpara)$}

    \put(32,22){Rotation by $\U$}
    \end{overpic}
    }
    \caption{
        Schematic of the inference with the \uncons and \cons analysis maps $\tmap$ and $\constmap$ in \cref{eqn:analysis_map_block,eqn:constrained_analysis_map}.
        \textcolor{LightBlue}{For the \uncons map $\tmap$}, we (i) construct $\tmap$ and (ii) \textcolor{LightBlue}{apply $\tmap$ to get the update $\x_a = \tmap(\y, \x)$}. 
        \textcolor{MainOrange}{For the \cons map $\constmap$}, we
        (i) \textcolor{AccentOrange}{project $\x$ onto the columns of $\Uperp$} \textcolor{MainGreen}{and $\Upara$}, 
        (ii) construct $\constmap^{\parallel}$, 
        (iii) \textcolor{MainGreen}{apply $\constmap^{\parallel}$ to get the update $\x_{\parallel, a} = \constmap^{\parallel}(\y, \xperp, \xpara)$ in the span of $\Upara$},  
        and (iv) \textcolor{MainOrange}{lift the result to the original space}. 
        Notably, the \textcolor{LightBlue}{update of the \uncons map $\tmap$} is outside the \textcolor{MainBlue}{invariant-preserving subspace}, while the \textcolor{MainOrange}{update for the \cons map $\constmap$} remains on the \textcolor{MainBlue}{invariant-preserving subspace}.
    }
    \label{fig:constrainedDA}
\end{figure}

Algorithm \ref{algo:nonlinear_constrained} presents pseudo-code for one analysis step of the \cons SMF (\consSMF). 
This algorithm transforms a set of forecast samples into a set of filtering samples by assimilating the observational data $\ystar$ into the prior samples $\{\x^i\}_{i=1}^M$ while preserving linear invariants.  
Moreover, Figure \ref{fig:constrainedDA} illustrates how the \uncons and \cons analysis maps operate.
Finally, see \citep{spantini2022coupling} for more details on the estimation of the \kr---underpinning the SMF used in this work---that pushes forward the joint forecast distribution $\pdf{(\Y_t, \X_t) \given \Y_{1:t-1} = \ystar_{1:t-1}}$ to the standard Gaussian distribution. 

A few remarks are in order. 

\begin{remark}[Exact Bayesian updates preserve invariants]
	Suppose the dynamical system's invariants are constant over the support of the prior. 
	In that case, the Bayesian update guarantees that the invariants are also constant over the posterior (see \ref{apx:invariants}). 
	In this setting, violations of the invariants are thus entirely due to the approximate treatment of the inference problem. 
	In this work, we are thus interested in approximate inference algorithms that preserve this critical property of Bayes' rule.
\end{remark}

\begin{remark}[Empirical estimators of \linpam s preserve linear invariants]\label{remark:empirical_linpam}
	Let $\widehat{\BB{T}}_{\ystar}^{\parallel}$ be an empirical estimator for the constrained analysis map $\constmap^{\parallel}$ of \eqref{eqn:constrained_tmap_para}, and $\widehat{\BB{T}}_{\ystar}$ be the resulting estimator of the \linpam\ $\constmap$ of \eqref{eqn:constrained_analysis_map}. 
	In practice, the estimator $\widehat{\BB{T}}_{\ystar}^{\parallel}$ is imperfect as we estimate it from samples $\{(\iup{\y}, \iup{\x})\}$ of an approximation of the joint forecast distribution $\pdf{(\Y_t, \X_t) \given \Y_{1:t-1} = \ystar_{1:t-1}}$, and the realization $\ystar_t$ to assimilate can originate from an approximation of the observation distribution $\pdf{\Y_t}$. 
	Independent of these discrepancies and the quality of the estimator, $\widehat{\BB{T}}_{\ystar}$ preserves linear invariants. 
	To prove this, we introduce the class of transformations $ \mathcal{T}_{\text{LinP}} = \{ \BB{T} \colon \real{d} \times \real{n} \to \real{n} \given \BB{T}(\y, \x) = \Uperp \Uperp^\top \x  + \Upara \BB{K}(\y, \x) \}$ with $\BB{K}(\y, \x) = \widehat{\BB{T}}_{\ystar}^{\parallel}(\y, \Uperp^\top \x, \Upara^\top \x)$. 
	By noting that any estimator $\widehat{\BB{T}}_{\ystar}$ belongs to $\mathcal{T}_{\text{LinP}}$ and that elements of $\mathcal{T}_{\text{LinP}}$ preserve the linear invariants $\invar(\x) = \Uperp^\top \x$, we obtain the desired result.
\end{remark}

\begin{remark}[Treatment of constant linear invariants over the prior]\label{remark:constant_prior}
Assume that the linear invariants are constant over the prior $\pdf{\X}$, \ie $\Uperp^\top \x = \C \in \real{r}$ for any realization $\x \in \real{n}$ of $\X$. 
Thus, the prior $\pdf{\X}$ is not supported on $\real{n}$ but on the affine space $\{ \x \in \real{n} \given \Uperp^\top \x = \C \}$. 
This remark explains how to adapt the construction of \linpam \ of \eqref{eqn:constrained_analysis_map} to this setting. 
Let $\x \in \real{n}$ be a realization of $\X$. 
From the state decomposition \eqref{eqn:state_decomposition}, we have $\x = \Uperp \C + \Upara \xpara$ with $\xpara \in \real{n-r}$. 
The distribution $\pdf{\Uperp^\top \X}$ becomes a point mass centered at $\C$, \ie $\pdf{\Uperp^\top \X}(\xperp) = \pdf{\Xperp}(\xperp) = \delta(\C - \xperp)$, where $\delta$ denotes the Dirac delta distribution. 
Thus, we can omit the rotated state variable $\Xperp$ from the analysis and operate on the reduced joint space $(\Y, \Upara^\top \X) = (\Y, \Xpara)$. 
We define the reduced map $\smap^{\Xparaup}_{\text{reduced}}$ pushing forward $\pdf{\Y, \Xpara}$ to $\eta_{\Xpara}$ as  $\smap^{\Xparaup}_{\text{reduced}}\colon \real{d} \times \real{n -r} \to \real{n-r}, \; (\y, \xperp) \mapsto \smap^{\Xparaup}_{\text{reduced}}(\y, \xpara)$. 
We stress that we don't need to define a reduced map $\smap^{\Xperpup}_{\text{reduced}}$ and that $\smap^{\Xparaup}_{\text{reduced}}$ does not depend on the rotated state coordinates $\xperp$. The reduced constrained analysis map $\widetilde{\BB{T}}^\parallel_{\ystar, \text{reduced}} \colon \real{d} \times \real{n -r} \to \real{n-r}$ reads $\widetilde{\BB{T}}^\parallel_{\ystar, \text{reduced}}(\y, \xpara) = \smap^{\Xparaup}_{\text{reduced}}(\ystar, \cdot)^{-1} \circ \smap^{\Xparaup}_{\text{reduced}}(\y, \xpara)$. 
Finally, $\widetilde{\BB{T}}_{\ystar, \text{reduced}} \colon \real{d} \times \real{n} \to \real{n}$ formulated in the original space becomes
\begin{equation}
\label{eqn:constrained_analysis_map_constant}
        \widetilde{\BB{T}}_{\ystar, \text{reduced}}(\y, \x)  =  \Uperp \C  + \Upara \widetilde{\BB{T}}^\parallel_{\ystar, \text{reduced}}(\y, \Upara^\top \x).
\end{equation}
Thus, the state update is confined to the affine space $\{ \x \in \real{n} \given \Uperp^\top \x = \C \}$, as we only update the state coordinates $\Upara^\top \x$ spanning this subspace.
\end{remark}
\section{Specialization to the Gaussian case \label{sec:gaussian}}

We now specialize the \cons analysis map $\constmap$ of \eqref{eqn:constrained_analysis_map} for the Gaussian case. 
Let $\Y \in \real{d}$ and $\X \in \real{n}$ be random variables that are jointly Gaussian, i.e.,
\begin{equation}
    \stack{\Y}{\X} \sim \N{\stack{\meanx}{\meany}}{
    \begin{bmatrix}
    \covy & \cov{\X, \Y}^\top \\
    \cov{\X, \Y} &  \covx
    \end{bmatrix}}.
\end{equation}
We define the change of coordinates $\widetilde{\U}$ transforming $(\Y, \Xperp)$ into $(\Y, \Xperp, \Xpara)$ as
\begin{equation}
    \widetilde{\U}^\top = \begin{bmatrix}
        \id{d} & \zero{d \times n}\\
        \zero{d \times n} & \U^\top
    \end{bmatrix} = \begin{bmatrix}
        \id{d} & \zero{d \times n}\\
        \zero{d \times r} & \Uperp^\top\\
        \zero{d \times n-r} & \Upara^\top
        \end{bmatrix}
\end{equation}
and note that $(\Y, \Xperp, \Xpara)$ is also Gaussian with statistics
\begin{equation}
     \begin{bmatrix}
         \Y \\
         \Xperp\\
         \Xpara
     \end{bmatrix} = \widetilde{\U}^\top \YX \sim \N{\begin{bmatrix}
         \meany \\
         \mean{\Xperp}\\
         \mean{\Xpara}
     \end{bmatrix}}{
    \begin{bmatrix}
    \covy & \cov{\Xperp, \Y}^\top &  \cov{\Xpara, \Y}^\top \\
    \cov{\Xperp, \Y} & \cov{\Xperp} &  \cov{\Xperp, \Xpara}^\top\\
    \cov{\Xpara, \Y} & \cov{\Xperp, \Xpara} &  \cov{\Xpara}
    \end{bmatrix}}.
\end{equation}
For a generic $m$-dimensional Gaussian random variable $\Z \sim \N{\meanz}{\covz}$, we denote by $\lmap_{\Z} \in \real{m \times m}$ the lower Cholesky factor of its inverse covariance matrix, i.e., $\cov{\Z}^{-1} = \lmap_{\Z}^\top \lmap_{\Z}$. 
To construct the \kr $\smap$ that pushes forward $\pdf{\Y, \Xperp, \Xpara}$ to the standard Gaussian distribution, we introduce the Cholesky factors $\lmap_{\Y} \in \real{d \times d}, \lmap_{\Xperp \given \Y} \in \real{r \times r}, \lmap_{\Xpara \given \Y, \Xperp} \in \real{(n-r) \times (n-r)}$ with $\covy^{-1} = \lmap_{\Y}^\top \lmap_{\Y}$, $\cov{\Xperp \given \Y}^{-1} = \lmap_{\Xperp \given \Y}^\top \lmap_{\Xperp \given \Y}$, and $\cov{\Xpara \given \Y, \Xperp}^{-1} = \lmap_{\Xpara \given \Y, \Xperp}^\top \lmap_{\Xpara \given \Y, \Xperp}$, respectively. 
Using classical results on the conditional mean of Gaussian random variables, we obtain the \kr $\smap$ as 
\begin{equation}
\label{eqn:kr_gaussian}
    \smap(\y, \xperp, \xpara) = \left[\begin{array}{c}
\begin{aligned}
        & \smap^{\Yup}(\y) \\
        & \smap^{\Xperpup}(\y, \xperp) \\
        & \smap^{\Xparaup}(\y, \xperp, \xpara)
\end{aligned}
\end{array}\right] = \left[\begin{array}{c}
\begin{aligned}
        & \lmap_{\Y}(\y - \meany) \\
        & \lmap_{\Xperp \given \Y} \left(\xperp - \mean{\Xperp \given \Y}  \right) \\
        & \lmap_{\Xpara \given \Y, \Xperp}(\xpara - \mean{\Xpara \given \Y, \Xperp})
\end{aligned}
\end{array}\right].
\end{equation}
By specializing the formulas \eqref{eqn:tmap_perp} and \eqref{eqn:tmap_para} to the triangular map $\smap(\y, \xperp, \xpara)$ of \eqref{eqn:kr_gaussian}, we obtain
\begin{equation}
\label{eqn:gaussian_analyssis_map_decoupled}
    \begin{aligned}
        \tmap^{\perp}(\y,  \xperp) 
            & = \xperp - \cov{\Xperp, \Y} \covy^{-1}(\y - \ystar), \\
        \tmap^{\parallel}(\y,  \xperp, \xpara) 
            & = \xpara - \cov{\Xpara, \Y} \covy^{-1}(\y - \ystar).
    \end{aligned}
\end{equation}
Note that the formulas for $\tmap^{\perp}$ and $\tmap^{\parallel}$ correspond to the Kalman filter update in the subspaces spanned by the columns of $\Uperp$ and $\Upara$, respectively. 
Remarkably, despite the recursive update of the state components (first updating $\xperp$ then $\xpara$) in the generic nonlinear case, the analysis map $\tmap^{\parallel}(\y,  \xperp, \xpara)$ in the Gaussian case does not depend on $\xperp$. 
This suggests that we can fully decouple the update of the different state components with the Kalman filter, echoing a similar conclusion in the context of Kalman smoothers by \cite{ramgraber2023_smoothing_part1}. 
The analysis map in the original space $\tmap$ now reads
\begin{equation}
\label{eqn:gausssian_analysis_map}
\begin{aligned}
       \tmap(\y, \x) & =  \Uperp \tmap^{\perp}(\y,  \xperp) + \Upara \tmap^{\parallel}(\y,  \xperp, \xpara)\\
                  & =  \Uperp \left(\xperp - \cov{\Xperp, \Y} \covy^{-1}(\y - \ystar) \right) + \Upara (\xpara - \cov{\Xpara, \Y} \covy^{-1}(\y - \ystar)) \\
                  & = \Uperp (\Uperp^\top \x)  + \Upara (\Upara^\top \x) -  \Uperp \Uperp^\top \cov{\X, \Y} \covy^{-1}(\y - \ystar) - \Upara \Upara^\top \cov{\X, \Y} \covy^{-1}(\y - \ystar)\\
                  & = \x - \cov{\X, \Y}\covy^{-1}(\y - \ystar), 
\end{aligned}
\end{equation}
where the last two equations rely on $\xperp = \Uperp^\top \x$ and $\U = [\Uperp, \Upara]$ being orthonormal. 
As noted in the previous section, rotating the state components, performing inference in the new coordinates, and lifting the result back to the original space is strictly equivalent to performing inference in the original space. 
Thus, it is natural that the analysis map of \eqref{eqn:analysis_map_block} in the Gaussian case reverts to the Kalman filter update. Similarly, we obtain the \cons analysis map $\constmap$ preserving the invariant $\invar(\x) = \Uperp^\top \x$ as 
\begin{equation}
\label{eqn:constrained_gausssian_analysis_map}
\begin{aligned}
       \constmap(\y, \x) & =  \Uperp \xperp + \Upara \tmap^{\parallel}(\y,  \xperp, \xpara)\\
                  & = \x - \Upara\Upara^\top \cov{\X, \Y}\covy^{-1}(\y - \ystar)\\
                  & = \x - (\id{} - \Uperp\Uperp^\top) \cov{\X, \Y}\covy^{-1}(\y - \ystar).
\end{aligned}
\end{equation}
We conclude this section with several comments on this last result. 

\begin{remark}
	First, the \cons analysis map corresponds to a Kalman-like update where the Kalman gain is projected on the orthonormal complement of $\Uperp$, namely on the space spanned by the columns of $\Upara$. 
	Interestingly, \eqref{eqn:constrained_gausssian_analysis_map} establishes an equivalence between an ``embedding approach'' and a ``projective approach'' \citep{simon2010kalman}. 
	In the first case, we perform a \cons inference in the rotated space before lifting the result back. 
	In the latter case, we ``naively'' project the state's update $-\cov{\X, \Y}\covy^{-1}(\y - \ystar)$ onto the columns of $\Upara$ to ensure that the invariant is unchanged. 
	While there is no reason for these two approaches to coincide in the non-Gaussian case with arbitrary invariants, \eqref{eqn:constrained_gausssian_analysis_map} states that they are equivalent in the Gaussian case with linear invariants $\invar$. 
\end{remark}

\begin{remark}
	In the case where the linear invariants are constant over the prior---see Remark \ref{remark:constant_prior} for the general treatment---we can specialize the reduced transformations of Remark \ref{remark:constant_prior} to derive a reduced constrained analysis map $\widetilde{\BB{T}}_{\ystar, \text{reduced}}$ of \eqref{eqn:constrained_analysis_map_constant} in the Gaussian case. 
	One can show that the resulting map is identical to the fully supported case, since the unconstrained map $\tmap^{\parallel}(\y,  \xperp, \xpara)$ of \eqref{eqn:gaussian_analyssis_map_decoupled} does not depend on $\xperp$. 
	However, for brevity, we omit details.
\end{remark}
\section{Preserving linear invariants with Kalman filters \label{sec:invariants_kalman}} 

It is often claimed that the vanilla Kalman filter and its Monte Carlo approximation, namely the ensemble Kalman filter, preserve linear invariants. 
This section addresses this claim---and its important nuances---in two particular scenarios. 
In the first scenario, we assume that the linear invariants are constant over the prior distribution $\pdf{\X}$, \ie $\invar(\x) = \C \in \real{r}$ for any realization $\x$ of $\X$. 
In the second scenario, we assume that the expected linear invariants of the prior distribution are known, \ie $\E{\pdf{\X}}{\invar(\x)} = \C \in \real{r}$. 
This can model a scenario in which the true invariant is known with some uncertainty. 

In the first case, where the linear invariants are constant over the prior distribution, we show below that the ensemble Kalman filter preserves linear invariants when \emph{no regularization} is applied.
Importantly, the empirical Kalman gain often needs to be regularized with covariance tapering or inflation to alleviate undesired sampling effects such as rank deficiency, spurious long-range correlations, and underestimation of the statistics.
However, such modifications are known to break the linear invariant-preserving properties of the vanilla ensemble Kalman filter \citep{janjic2014conservation}. 
We present a specific numerical example where the regularized \enkf (with inflation
and covariance tapering) does not preserve linear invariants in Section \ref{subsec:linad} for a linear advection equation.
(See Figure \ref{fig:linadvection_mass}.)

In the second scenario, where the expected linear invariants of the prior distribution are known, the ensemble Kalman filter no longer necessarily preserves linear invariants.
In contrast, the \cons Kalman filter and the \cons ensemble Kalman filter, presented in the previous section, intrinsically preserve linear invariants in both of the above scenarios.

\subsection{Scenario 1: The linear invariants are constant over the prior distribution}
\label{subsec:constant_scenario}

Assume that the linear invariants are constant over the prior distribution $\pdf{\X}$, say $\invar(\x) = \C \in \real{r}$ for any realization $\x$ of $\X$. 
This assumption implies $\E{\pdf{\X}}{\invar(\x)} = \C$. 
Recall that the analysis map of the Kalman filter $\tmapkf$ is given in \eqref{eqn:gausssian_analysis_map} as $\tmapkf(\y, \x) = \x - \cov{\X, \Y} \cov{\Y}^{-1}(\y - \y^\star)$, where $\cov{\X, \Y}$ denotes the cross-covariance between the state $\X$ and the observation $\Y$. 
Hence, we get 
\begin{equation}\label{eq:sce1}
\resizebox{.9\textwidth}{!}{$\displaystyle 
\begin{aligned}
    \invar(\tmapkf(\y, \x)) 
    & = \Uperp^\top \tmapkf(\y, \x) \\ 
    & = \Uperp^\top \x -  \Uperp^\top \cov{\X, \Y} \cov{\Y}^{-1}(\y - \y^\star) \\
    & = \Uperp^\top \x - \Uperp^\top \left[ \int \left(\x' - \mean{\X} \right)(\y' - \mean{\Y})^\top\d \pdf{\Y, \X}(\y', \x')  \right]\cov{\Y}^{-1}(\y - \ystar) \\
    &  = \Uperp^\top \x -  \left[ \int  \Uperp^\top \left(\x' - \mean{\X} \right) \pdf{\X}(\x') \left[ \int  \pdf{\Y \given \X}(\y' \given \x')  (\y' - \mean{\Y})^\top\d \y' \right] \d \x'\right] \cov{\Y}^{-1}(\y - \ystar).
\end{aligned}
$}
\end{equation}
Furthermore, we have 
$\Uperp^\top \x' \pdf{\X}(\x') = \C \pdf{\X}(\x')$ for any $\x' \in \real{n}$. 
Similarly, we have $\Uperp^\top \mean{\X} = \Uperp^\top \E{\pdf{\X}}{\x} = \E{\pdf{\X}}{\Uperp^\top \x} = \E{\pdf{\X}}{\invar(\x)} = \C$ and thus
$\Uperp^\top \mean{\X} \pdf{\X}(\x') = \C \pdf{\X}(\x')$ for any $\x' \in \real{n}$.
Hence, \eqref{eq:sce1} becomes  
\begin{equation}
\begin{aligned}
    \invar(\tmapkf(\y, \x)) 
    & = \Uperp^\top \x -  \left[ \int  \C \pdf{\X}(\x') - \C \pdf{\X}(\x')   \left[ \int  \pdf{\Y \given \X}(\y' \given \x')  (\y' - \mean{\Y})^\top\d \y' \right] \d \x'\right] \cov{\Y}^{-1}(\y - \ystar) \\
    & = \Uperp^\top \x - \zero{r} \\
    &  = \invar(\x).
\end{aligned}
\end{equation}
This implies that the invariant is preserved by $\tmapkf$ if the invariant is constant over $\pdf{\X}$. 
We note that a similar derivation holds for the analysis map of the \textit{ensemble} Kalman filter, where the covariance matrices $\cov{\X, \Y}$ and $\cov{\Y}$ are replaced by their empirical counterparts estimated from samples $\{ (\iup{\x}, \iup{\y}) \}$. 
Finally, we stress that this result holds for any joint distribution $\pdf{\Y, \X}$ with finite second-order moments.

\subsection{Scenario 2: The expected value of the linear invariants over the prior is known}

We now make the weaker assumption that the prior \textit{expectations} of the invariants, i.e., $\E{\pdf{\X}}{\invar(\x)} = \C \in \real{r}$, are known. 
We investigate whether the Kalman filter preserves the expected linear invariants by checking if $\E{\pdf{\Y, \X}}{\invar(\tmapkf(\y, \x))} = \C$ is satisfied. 
To this end, note that \eqref{eqn:gausssian_analysis_map} implies
\begin{equation}\label{eq:sce2}
\begin{aligned}
    \E{\pdf{\Y, \X}}{\invar(\tmapkf(\y, \x))} 
        & = \E{\pdf{\Y, \X}}{\Uperp^\top \x} -  \E{\pdf{\Y, \X}}{\Uperp^\top \cov{\X, \Y} \cov{\Y}^{-1}(\y - \y^\star)} \\
        & = \C - \Uperp^\top \cov{\X, \Y} \cov{\Y}^{-1} (\mean{\Y} - \ystar).
\end{aligned}
\end{equation}
Notably, the second term in \eqref{eq:sce2} is not necessarily zero.
Thus, the Kalman filter does generally not preserve the expected linear invariants of the prior $\pdf{\X}$. 
To be more precise, \eqref{eq:sce2} implies that the Kalman filter preserves the expected linear invariants of the prior $\pdf{\X}$ if and only if $\Uperp^\top \BB{b} = \zero{r}$ with ``expected state update'' $\BB{b} = \cov{\X, \Y} \cov{\Y}^{-1} (\mean{\Y} - \ystar) \in \real{n}$.
We can interpret this as $\BB{b}$ lying in the span of the columns of $\Upara$ and, therefore, not containing any information about the linear invariants.

\section{Computational examples}
\label{sec:examples}

We provide a series of numerical examples on the preservation of linear invariants with ensemble Kalman filters and SMFs. 
We consider (i) a synthetic linear ordinary differential equation (ODE) with an arbitrary number of linear invariants, (ii) the linear advection equation as a simple prototype of a hyperbolic conservation law, and (iii) an embedding of the nonlinear \lo in $\real{4}$ with a linear invariant. 
For reproducibility, the code of our numerical experiments is available at \href{https://github.com/mleprovost/Paper-Linear-Invariants-Ensemble-Filters}{\textcolor{magenta}{https://github.com/mleprovost/Paper-Linear-Invariants-Ensemble-Filters}}.

\subsection{Numerical setup}
\label{subsec:numerical_setup}

This section discusses common considerations in our data assimilation experiments. To isolate the performance of the filtering algorithms, we perform a ``twin experiment'' in which the same state-space model is used to generate the ground truth and in the forecast and analysis steps of the ensemble filters \citep{asch2016data}. 
The ground truth is generated by sampling an initial state $\x_0^\star$ from the initial distribution $\pdf{\X_0}$ and evolving it through the dynamical model \eqref{eqn:dyn} over the time interval $[0, t_f]$. 
Discretizing $[0, t_f]$ using an equidistant grid with step size $\dtobs$, at each time step, we collect the true state $\x^\star_t$ and generate a noisy observation $\ystar_t$ from the observation model \eqref{eqn:obs}. 
Although we might not know the transition kernel $\pdf{\X_{t} \given \X_{t-1}}$ and the likelihood model $\pdf{\Y_{t} \given \X_{t-1}}$, we assume that we can generate samples from them. 
Furthermore, we assume that the observation operator $\obs$ is linear and given by $\obs(\x_t) = \Obs \x_t$ with observation matrix $\Obs \in \real{d \times n}$. 
In the filtering setting that underpins this work, we seek to sequentially estimate the posterior of the states $\{\x^\star_t \}$ given noisy observations $\{ \ystar_t \}$ of the true process and the initial distribution $\pdf{\X_0}$. 

To ensure that the process noise $\Noisedyn_t$ in the dynamical model \eqref{eqn:dyn} does not change the linear invariants $\invar(\x) = \Uperp^\top \x$, we consider linear invariant-preserving Gaussian process noise with zero mean and covariance $\sigma_{\Noisedyn_t}^2 \id{n}$. 
To sample from this distribution, we project samples from $\N{\zero{n}}{\sigma_{\Noisedyn_t}^2 \id{n}}$ on the columns of $\Upara$ spanning the (orthogonal) complement of $\Span\{\Uperp\}$. 
We stress that the numerical schemes used in our computational experiments preserve the dynamical systems' linear invariants. 
Thus, violations of these invariants can be entirely attributed to flaws in the filtering algorithms.

We assess the performance of the different ensemble filters using the root-mean-square error (RMSE) and the spread of the posterior ensemble. 
Recall that the RMSE at time $t$ is defined as $\rmse_t^2 = ||\x^\star_t - \widehat{\BB{\mu}}_{t, a}||_2^2/n$, where $\widehat{\BB{\mu}}_{t, a} \in \real{n}$ is the filtering ensemble mean at time $t$, and $n$ is the dimension of the state variable. 
The spread at time $t$ is given by $\spread_t^2 = \operatorname{tr}(\scov{t,a})/n$, where $\scov{t,a} \in \real{n \times n}$ is the filtering ensemble covariance at time $t$. 
In all experiments, the filters are run over $2 \cdot 10^3$ assimilation cycles, and we discard the first $10^3$ cycles to ensure the statistics are approximately stationary. 
After this spin-up phase, we report the time-averaged metrics of the ensemble filters with optimally tuned multiplicative inflation and/or covariance tapering, achieving the lowest RMSE for a given ensemble size.

\subsection{A synthetic linear model with an arbitrary number of linear invariants}
\label{subsec:synthetic_linear}

We start by considering a linear Gaussian filtering problem, for which we compare two related ensemble filters: the stochastic ensemble Kalman filter \citep{evensen1994sequential} without constraints on the linear invariants, called \uncons \enkf (\unenkf), and the proposed stochastic ensemble Kalman filter that preserves linear invariants by building a Monte Carlo approximation of the analysis map \eqref{eqn:constrained_gausssian_analysis_map}, called \cons \enkf (\consenkf). 
For both filters, we apply multiplicative inflation and covariance tapering to regularize the empirical estimate of the Kalman gain, see \citep{asch2016data} for more details. 
We stress that these techniques are critical for successful state estimation with ensemble Kalman filters using limited samples, but they are also known to break linear invariants \citep{janjic2014conservation}. 
The proposed \cons \enkf combines the previous regularization techniques with a projection on the span of $\Upara$ to preserve linear invariants. 
Algorithm \ref{algo:linear_constrained} in \ref{apx:linear_constrained} provides a pseudo-code for the \cons \enkf.

Consider a synthetic linear ODE model for which we can fix an arbitrary number of linear invariants. 
This allows us to investigate the performance of the proposed \consenkf, compared to the existing \unenkf, for different numbers of invariants $r$ and ensemble sizes $M$. 
To this end, consider the forward model
\begin{equation}
\label{eqn:linear_model}
        \frac{\d \x}{\dt} = \A_r \x, \quad
        \x(0) = \x_0,
\end{equation}
where $\x_0 \in \real{n}$ and the matrix $\A_r \in \real{n \times n}$ is symmetric negative semidefinite. 
Its spectral decomposition is
\begin{equation}
\label{eqn:spectral}
    \A_r = \U \diag{\eigen_r} \U^{T},
\end{equation}
where $\U \in \real{n \times n}$ is an orthonormal matrix and $\diag{\eigen_r} = \operatorname{diag}(\eigen_r)$ contains the eigenavlues $\eigen_r \in \real{n}$ of $\A_r$. 
The dynamic model \eqref{eqn:linear_model} is motivated by production-destruction systems in atmospheric chemistry for which preserving linear invariants is crucial for accurate and robust numerical simulations \citep{huang2022stability,izgin2023stability,izgin2023study}.

Without loss of generality, we assume that the eigenvalues in $\eigen_r$ are ordered decreasingly as $\eigen_r = [\zero{r}, -\lambda_{r+1}, \ldots, -\lambda_{n}]$ with $\lambda_{k} > 0$ for $k>r$. 
We use the subscript $r$ to stress that $0$ is an eigenvalue with multiplicity $r$. We choose nonnegative eigenvalues for $\A_r$ to produce a linear dynamical model \eqref{eqn:linear_model} with stable solutions \citep{huang2022stability}. 
The solution of \eqref{eqn:linear_model} is 
\begin{equation}
    \x(t) 
    		= \exp(\A_r t) \x_0
		= \U \exp(\diag{\eigen_r} t) \U^\top \x_0.
\end{equation}
We denote by $\Uperp \in \real{n \times r}$ the first $r$ orthonormal columns of $\U$ and by $\Upara$ the remaining $(n-r)$ orthonormal columns. Thus, we can verify that the linear invariants $\Uperp^\top \x_0  = \C_0 \in \real{r}$ are preserved by \eqref{eqn:linear_model}:
\begin{equation}
\begin{aligned}
	\Uperp^\top \x(t) & = \Uperp^\top \U \exp(\diag{\eigen_r} t) \U^\top \x_0,\\
		& = \Uperp^\top \Uperp \diag{\one_{r}} \Uperp^\top \x_0 + \Uperp^\top \Upara \diag{[\exp(-\lambda_{r+1} t), \ldots, \exp(-\lambda_n t)]} \Upara^\top \x_0,\\
		& = \id{r}\C_0 + \zero{r},\\
		& = \C_0,
\end{aligned}
\end{equation}
where we have used that $\Uperp^\top \Uperp = \id{r}$, $\Uperp^\top \Upara = \zero{r \times (n-r)}$, and $\exp(\diag{\eigen_r} t) = \diag{[\one_{r}, \exp(-\lambda_{r+1} t), \ldots, \exp(-\lambda_n t)]}$. 
We denote by $\one_{r} \in \real{r}$ the vector of ones of length $r$.

From \eqref{eqn:linear_model}, we construct a discrete time forward operator $\dyn$ for the dynamical model \eqref{eqn:dyn} by integration of \eqref{eqn:linear_model} over a time step $\dtobs$ between two assimilation cycles, i.e., $\dyn(\x_t) = \exp(\A_r \dtobs) \x_t$.  
Furthermore, we use a Gaussian linear invariant-preserving process noise with zero mean and covariance $\sigma_{\Noisedyn_t}^2 \id{n}$ with $\sigma_{\Noisedyn_t} = 10^{-2}$. 
That is, the dynamic model \eqref{eqn:dyn} that propagates the state forward in time is 
\begin{equation}
    \X_{t+1} = \exp(\A_r \dtobs) \X_t + \Noisedyn_t, \; \text{for }t \geq 0
\end{equation} 
with $\Noisedyn_t \sim \mathcal{N}(\mathbf{0},\sigma_{\Noisedyn_t}^2 \id{n})$. 
As mentioned above, this setting allows a parametric study of the performance metrics as a function of the ratio of linear invariants $r/n \in [0, 1]$ (with state dimension $n$) and the ensemble size $M$. 
We consider a state of dimension $n = 20$. 
The eigenvalues $-\lambda_k$ for $k>r$ are independently drawn from a uniform distribution on $[-5,0]$. 
The time step size between two assimilation cycles is $\dtobs = 10^{-1}$.
We observe every component of the state, i.e., $d = n = 20$, corrupted by an additive Gaussian observation noise with zero mean and covariance $\sigma_{\Noiseobs}^2 \id{d}$ with $\sigma_{\Noiseobs} = 10^{-1}$. 

\begin{figure}[tb] 
    \centering
    \begin{subfigure}[b]{0.475\textwidth}
		\includegraphics[width=\textwidth]{%
			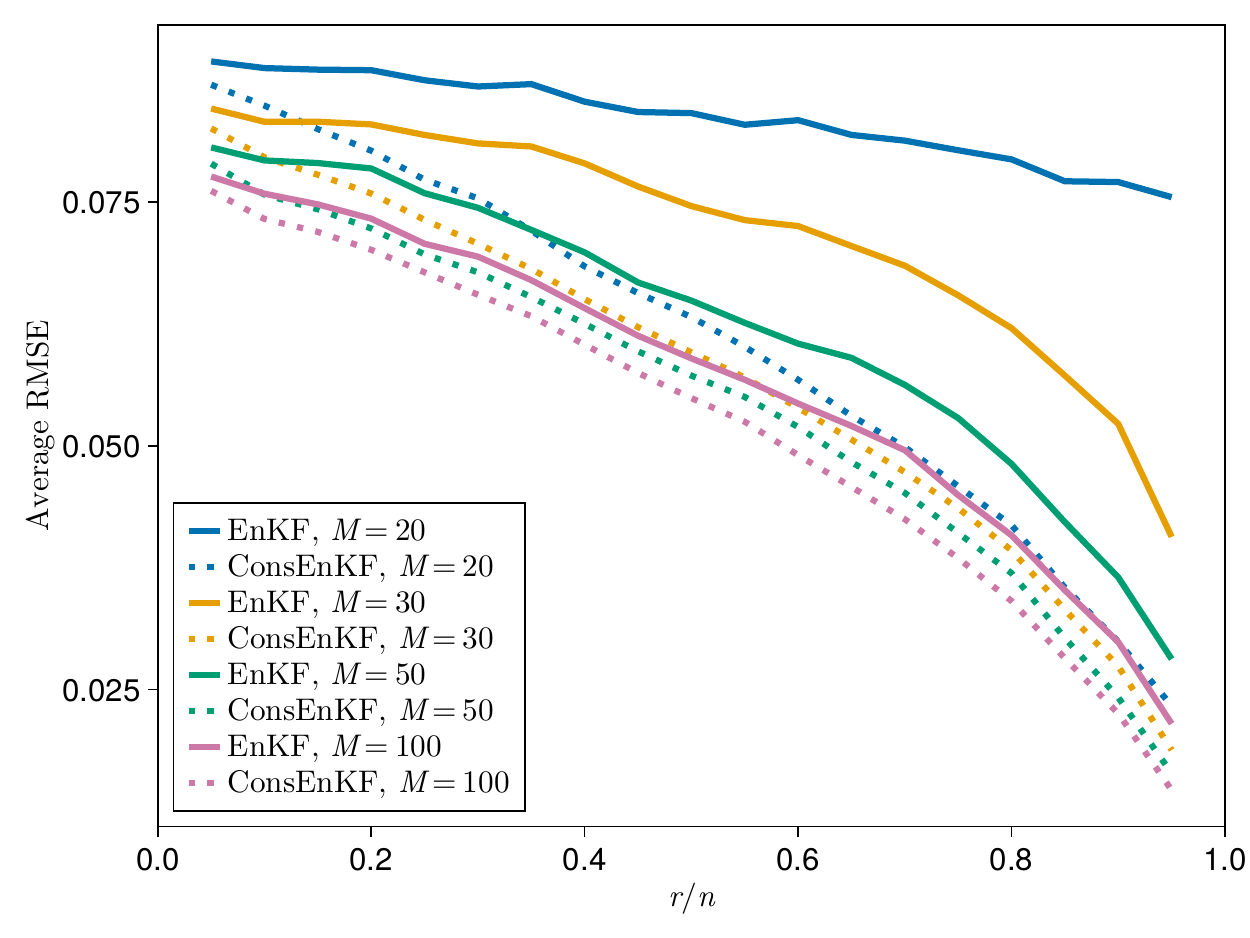} 
        \caption{RMSE for varying ratios $r/n$}
        \label{fig:toy_problem_RMSE_ratio}
    \end{subfigure}
    \begin{subfigure}[b]{0.475\textwidth}
		\includegraphics[width=\textwidth]{%
			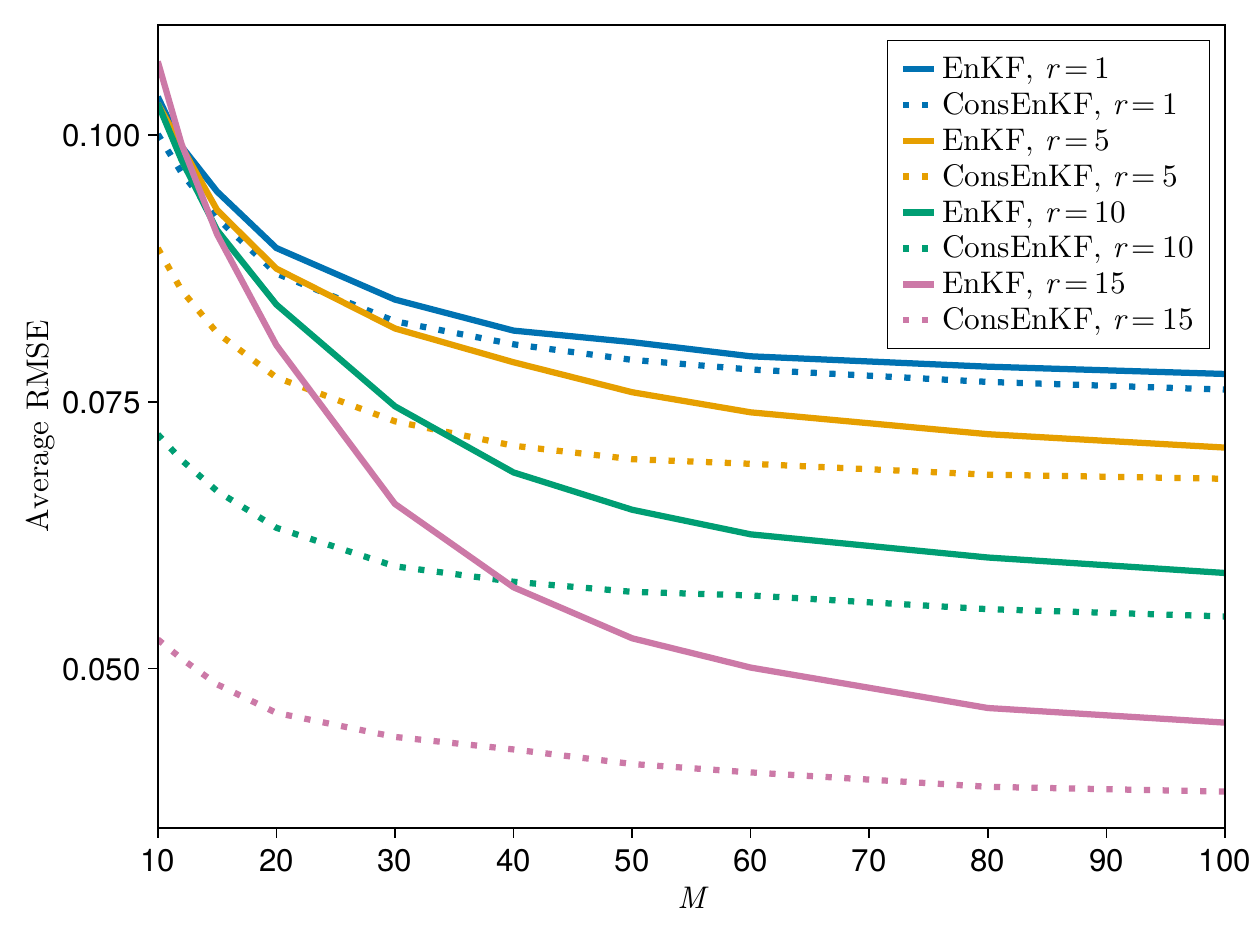} 
        \caption{RMSE for varying $M$}
        \label{fig:toy_problem_RMSE_M}
    \end{subfigure}
    \\ 
    \begin{subfigure}[b]{0.475\textwidth}
		\includegraphics[width=\textwidth]{%
			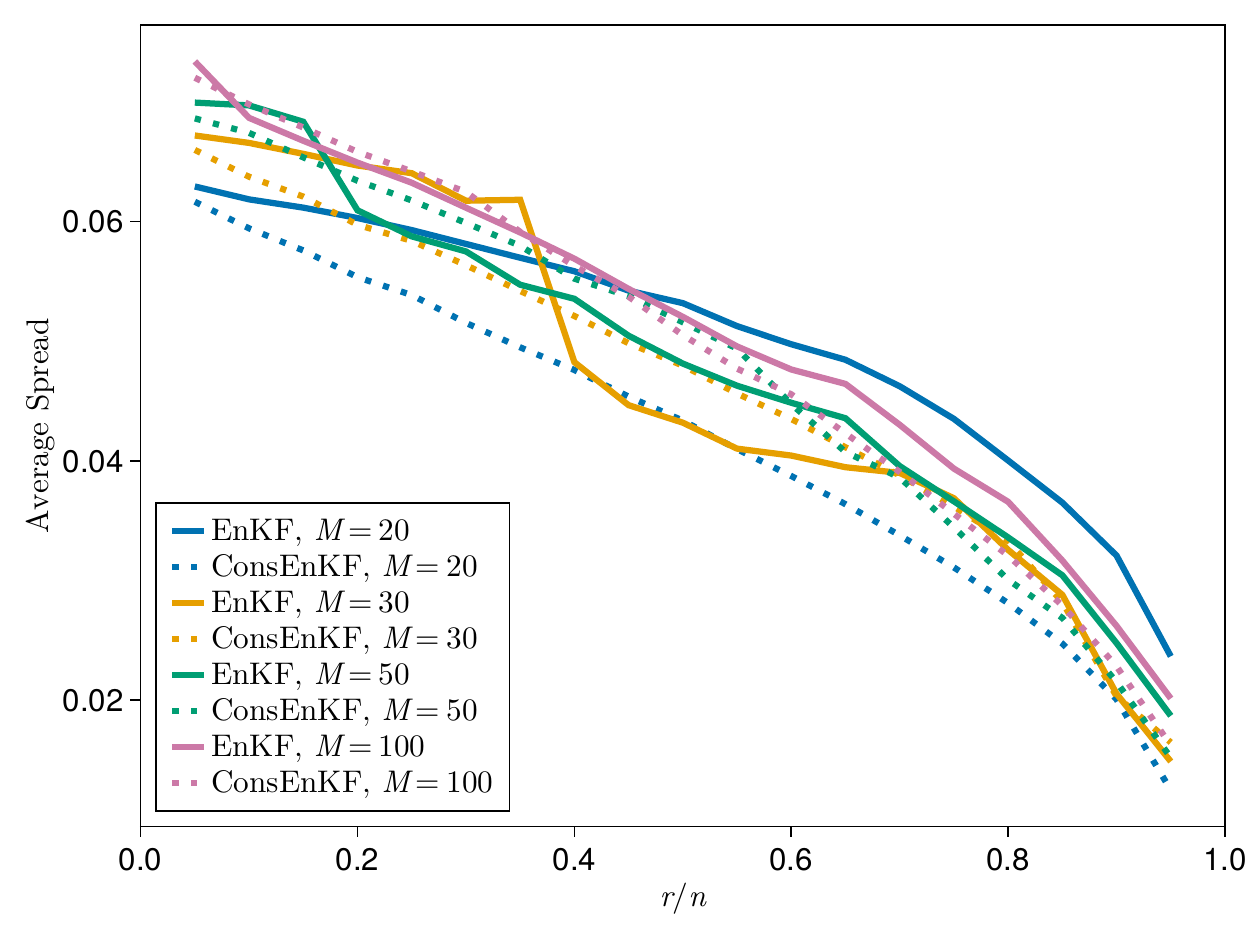} 
        \caption{Spread for varying ratios $r/n$}
        \label{fig:toy_problem_spread_ratio}
    \end{subfigure}
    \begin{subfigure}[b]{0.475\textwidth}
		\includegraphics[width=\textwidth]{%
			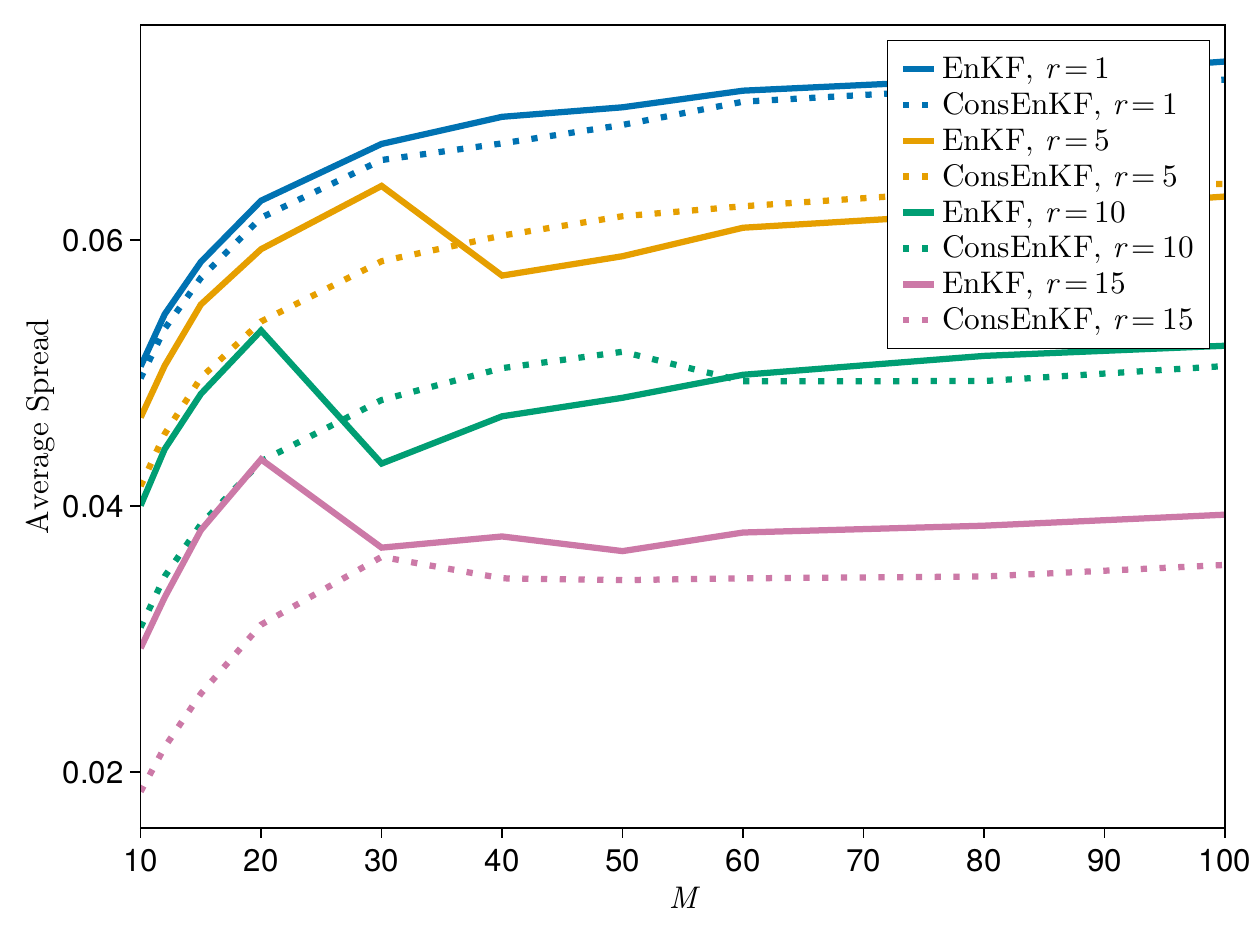} 
        \caption{Spread for varying $M$}
        \label{fig:toy_problem_spread_M}
    \end{subfigure}
    \caption{ 
    		Left: Time-averaged RMSE (top) and spread (bottom) for the \uncons \enkf (solid lines) and the \cons \enkf (dashed lines) for varying ratios $r/n \in [0,1]$ (left) between the number of linear invariants $r$ and the problem dimension $n$ with $M = 20, 30, 50, 100$ samples. 
    Right: Time-averaged RMSE (top) and spread (bottom) for the \uncons \enkf (solid lines) and the \cons \enkf (dashed lines) for varying ensemble sizes $M$ with $r = 1, 5, 10, 15$. 
    Both filters use optimally tuned multiplicative inflation and covariance tapering.
	}
    \label{fig:toy_problem}
\end{figure}

\cref{fig:toy_problem_RMSE_ratio,fig:toy_problem_RMSE_M} compare the RMSE of the \uncons \enkf and the \cons \enkf for different ensemble sizes $M$ and ratios $r/n$ between the number of linear invariants $r$ and the problem dimension $n$. 
Preserving the linear invariants consistently reduces the RMSE. 
For $r/n < 0.1$, the improvements are at around $5\%$, showing slight variation with the ensemble size $M$. 
As $r/n$ increases, we observe larger RMSE improvement for all ensemble sizes $M$. 
Specifically, for $M = 20$ and $r = 19$, the RMSE of the \uncons \enkf is $7.7 \cdot 10^{-2}$, while the RMSE of the \cons \enkf is $2.5 \cdot 10^{-2}$, corresponding to a reduction of the RMSE by $67\%$. 
Moreover, for small ensemble sizes $M < 40$, the RMSE improvements are significant and increase with the number of linear invariants. 
For $M = 10$ and $r = 10$, the RMSE is reduced by $36\%$. For a fixed number of invariants $r$, the RMSE improvement due to preserving linear invariants decreases with the ensemble size $M$. 
This is consistent with the empirical Kalman gain estimate requiring less regularization as the ensemble size $M$ increases. 
Thus, the analysis map of the \uncons \enkf gets closer to the analysis map of the vanilla \enkf. 
From Section \ref{subsec:constant_scenario}, the vanilla \enkf preserves linear invariants if all the forecast samples have the same invariants. 
We expect violations of the linear invariants of the \unenkf to decrease in magnitude as the ensemble size $M$ increases. 

Furthermore, \cref{fig:toy_problem_spread_ratio,fig:toy_problem_spread_M} report on the spread of the \uncons \enkf and the \cons \enkf. 
We do not observe significant differences in the spread of the \uncons \enkf and the \cons \enkf.
We conclude that preserving linear invariants is particularly advantageous for the RMSE when the ensemble size $M$ is small and the ratio of linear invariants $r/n$ is large. 
We interpret these results as a reduction of the variance of the empirical Kalman gain by constraining the image space of the estimated Kalman gain to the subspace spanned by the columns of $\Upara$. 
We refer readers to \citep{leprovost2022lowenkf} for further discussions on the benefits of linear dimension reductions for the ensemble Kalman filter.

\subsection{Linear advection on periodic domain}
\label{subsec:linad}

We next consider the one-dimensional \linad equation on the periodic domain $\domain = [0,1)$ as a forward model: 
\begin{equation}\label{eqn:linad}
\begin{aligned}
        \ddt{u(s, t)} + \div \left(c u(s, t) \right) & = 0, \quad && s \in \domain, \ t>t_0, \\
        u(s, t_0) & = u_0(s), \quad && s \in \Omega, 
\end{aligned}
\end{equation}
where $u(s, t)$ is the solution, $u_0 \colon \real{} \to \real{} $ is the initial state at time $t_0$, and $c = 1$ is the advection velocity. 
Notably, the total mass $m(t)  = \int_{\Omega} u(s,t) \d s$ of exact solutions $u(s,t)$ of \eqref{eqn:linad} is constant in time, i.e., $m(t) = m(t_0)$ for all $t$. 
Numerical solutions to \cref{eqn:linad} should mimic this behavior on a discrete level. 

We discretize  the domain $\domain = [0,1)$ with $n = 128$ grid nodes $\{\svec_k \}$. 
The state vector $\x_t$ at time $t$ is given by the pointwise evaluations of $u(s, t)$ at the grid nodes $\{\svec_j \}$, i.e., $\x_{t,k} = u(s_k, t)$ for $k = 1, \ldots, n$. 
To numerically solve \eqref{eqn:linad}, we use a spectral method for the spatial discretization and a fourth-order adaptive strong stability preserving Runge--Kutta (SSPRK43) time integration scheme \citep{brunton2019data, rackauckas2017differentialequations}. 
That is, the dynamic model \eqref{eqn:dyn} that propagates the state forward in time is 
\begin{equation}
    \X_{t+1} = \dyn(\X_t) + \Noisedyn_t, \; \text{for }t \geq 0,
\end{equation} 
where $\dyn(\X_t)$ corresponds to the numerical solution of \cref{eqn:linad} with initial data $\X_t$, computed as outlined above. 
We consider a Gaussian linear invariant-preserving process noise $\Noisedyn_t$ with zero mean and covariance $\sigma_{\Noisedyn_t}^2 \id{n}$ with $\sigma_{\Noisedyn_t} = 10^{-2}$. 
The time step size between two assimilation cycles is $\dtobs = 2 \cdot 10^{-1}$ (corresponding to half the convective time $t_c = L/c$, where $L$ is the length of $\domain$).

The distribution $\pdf{\X_0}$ used to generate the true initial condition and the initial ensemble members is given by an $n$-dimensional smooth and periodic distribution, denoted $\mathcal{S}_{n,r}(\Uperp, \C, \alpha)$. 
It is parameterized by a sub-unitary matrix $\Uperp \in \real{n \times r}$, a vector $\C \in \real{r}$, and a smoothing parameter $\alpha > 0$. 
By design, a random variable $\X \sim \mathcal{S}_n(\Uperp, \C, \alpha)$ satisfies $\Uperp^\top \X = \C$. 
To generate samples from the distribution $\mathcal{S}_n(\Uperp, \C, \alpha)$, we proceed in three steps: 
\begin{enumerate}
    \item[(i)]
    Generate samples $\z_{\operatorname{Re}}, \z_{\operatorname{Im}}$ from the $(n/2 + 1)$-dimensional standard Gaussian distribution $\N{\zero{}}{\id{}}$

    \item[(ii)] 
    Compute $\tilde{\x} \in \complex{\frac{n}{2} + 1}$ with components $\tilde{\x}_k = \left( \z_{\operatorname{Re}, k} + i \;  \z_{\operatorname{Im}, k} \right) \exp\left( -\frac{1}{2} k^\alpha \right)$, $k = 1, \ldots, \frac{n}{2} + 1$  

    \item[(iii)] 
    Compute $\x = \Uperp \C  + (\id{n} - \Uperp \Uperp^\top) \Finv (\tilde{\x})$
    
\end{enumerate}

Here, $i \in \complex{}$ is the imaginary unit and $\Finv \colon \complex{\frac{n}{2}+1} \to \real{n}$ denotes the inverse real fast Fourier transform, implemented under \code{rfft} in the FFTW library \citep{frigo1998fftw}. 
Figure \ref{fig:initial_ensemble} shows six samples $\x$ from this distribution. 

\begin{figure}[tb]
    \centering
    \includegraphics[width = 0.6\linewidth]{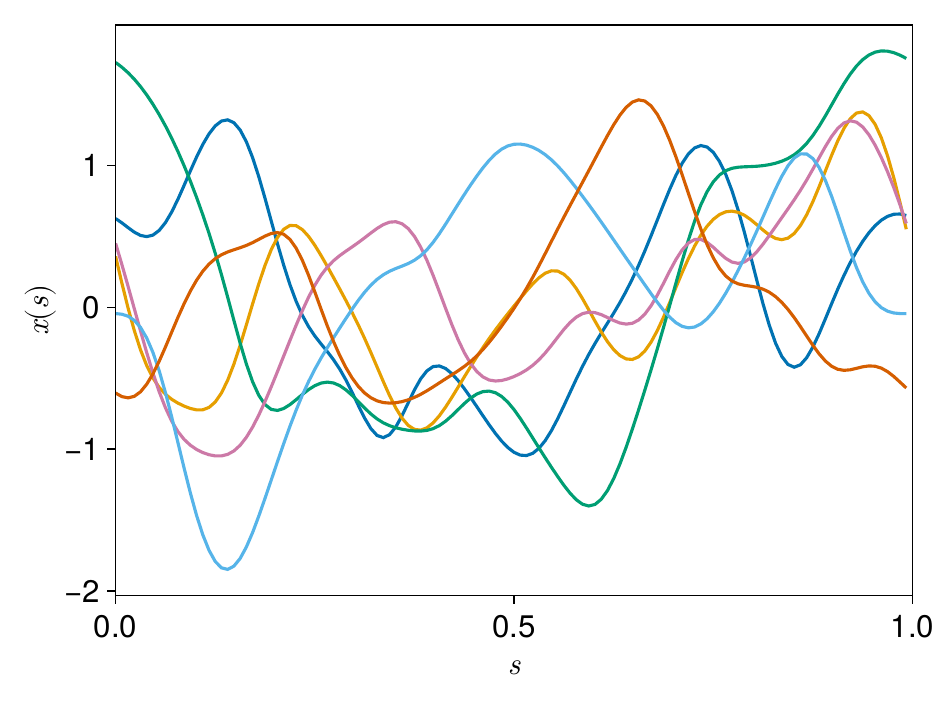}
    \caption{Six samples drawn from the distribution $\mathcal{S}_{n, r}(\Uperp, \C, \alpha)$ with $n = 128,\; r = 1, \; \Uperp = \one{}/n \in \real{n}, \C = 1 \in \real{},\; \text{and } \alpha = 1$. The linear constraint for each sample is $\Uperp^\top \x = 1$.}
    \label{fig:initial_ensemble}
\end{figure}

We have observations at every fourth grid point ($d = 32$) corrupted by an additive zero-mean Gaussian noise with covariance $\sigma_{\Noiseobs}^2 \id{d}$ with $\sigma_{\Noiseobs} = 10^{-1}$. 
The true mass $\C$ is drawn from a Gaussian distribution with mean $1$ and standard deviation $5 \cdot 10^{-2}$. 

\begin{figure}[tb]
    \centering
    \includegraphics[width = 0.475\linewidth]{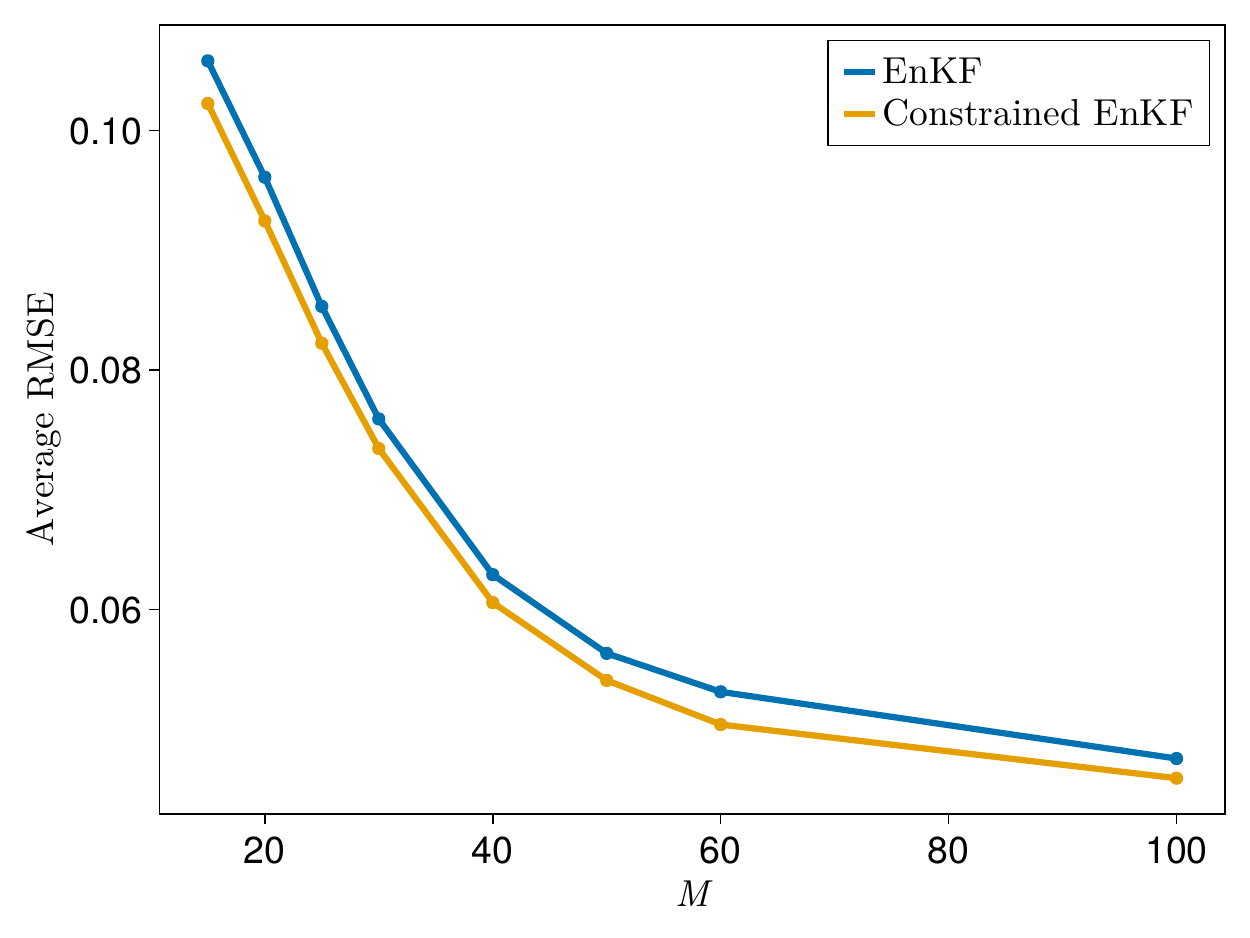}
    ~
    \includegraphics[width = 0.475\linewidth]{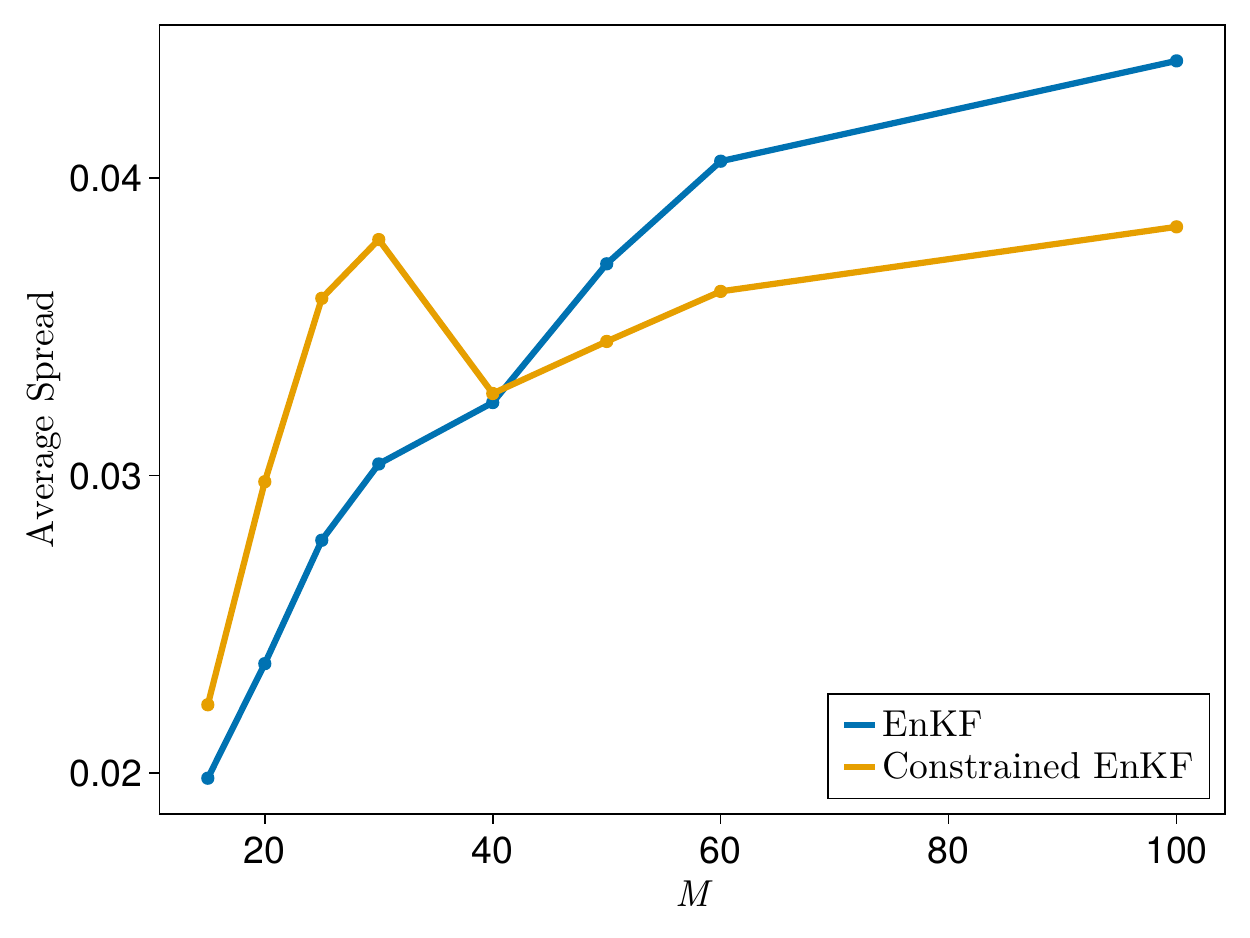}
    \caption{
        Left: Time-averaged RMSE of the \linad problem using the \unenkf (blue) and the \consenkf (yellow) for varying ensemble sizes $M$. 
        Right: Median spread for varying ensemble sizes $M$.
        Both filters use optimally tuned inflation and covariance tapering. 
    }
    \label{fig:linadvection_metrics}
\end{figure}

Figure \ref{fig:linadvection_metrics} reports the time-averaged evolution of the RMSE and the spread of the \linad problem using the \unenkf (blue) and the \consenkf (yellow) for varying ensemble sizes $M$. 
Both filters use optimally tuned inflation and covariance tapering.
We note that the \consenkf has a slightly better RMSE than its \uncons version. 
This result is consistent with the results of Section \ref{subsec:synthetic_linear} for small ratios $r/n$.
The spread over the ensemble size is similar. 
While these global metrics can be somewhat deceptive regarding improvement, the true benefit of preserving the mass is revealed by examining the evolution of the mass estimate, shown in Figure \ref{fig:linadvection_mass}.

\begin{figure}[tb]
    \centering
    \includegraphics[width = 0.6\linewidth]{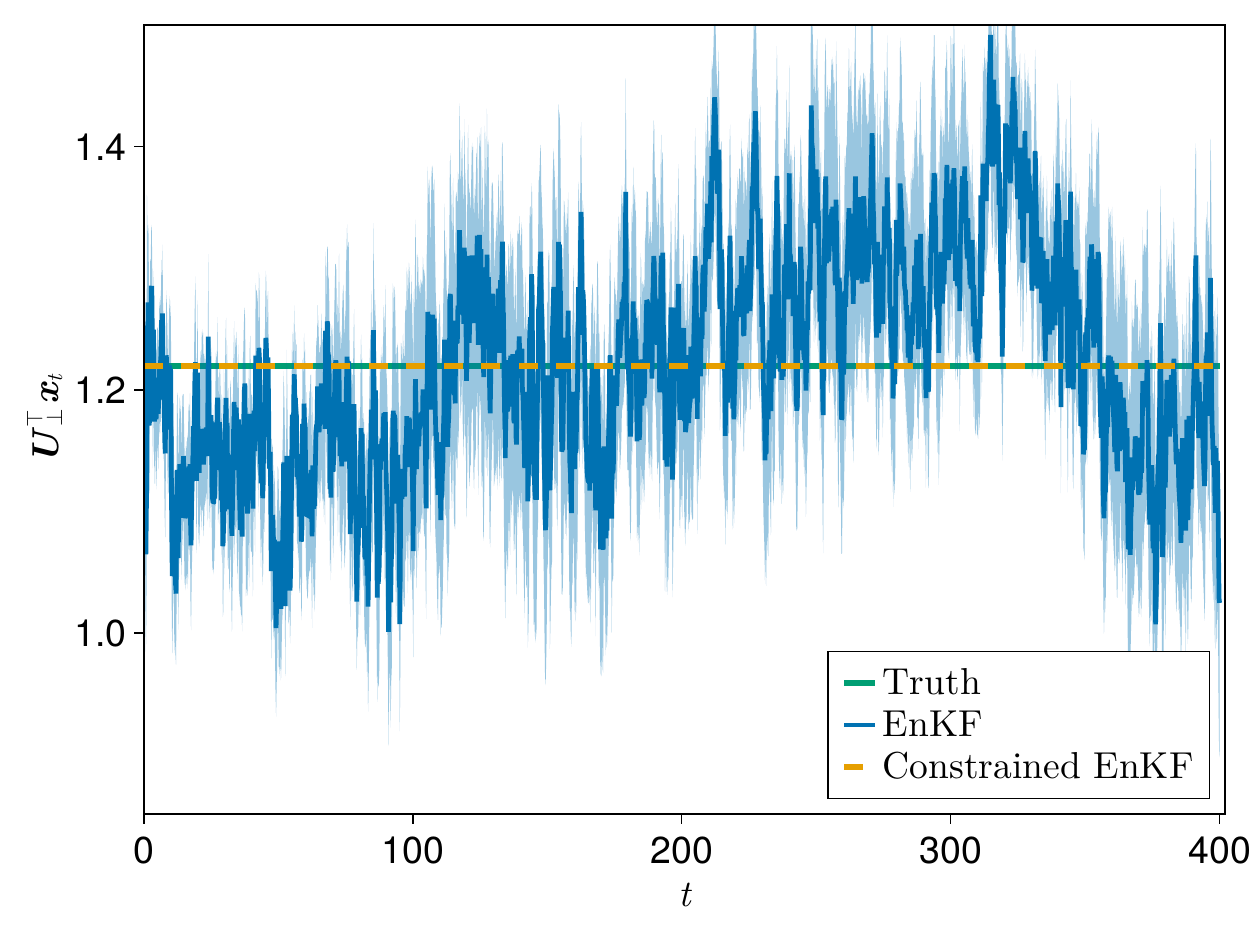}
    \vskip -0.5cm
    \caption{
        Time evolution of the linear invariant $\Uperp^\top \x_t$ for the true state process (green) and the posterior mean obtained with the \unenkf (blue) and the \consenkf (dashed yellow) for an ensemble size of $M = 40$. 
        Notably, $\Uperp^\top \x_t$ for the \consenkf and true state process align with each other. 
        Fainted areas show the $10\%$ and $90\%$ quantiles of the posterior estimate of the invariant. 
        Both filters use optimally tuned inflation and covariance tapering. 
    }
    \label{fig:linadvection_mass}
\end{figure}

We observe from Figure \ref{fig:linadvection_mass} that the \consenkf conserves mass up to machine precision. At the same time, the optimally tuned \unenkf shows notable unphysical variations in the mass estimate (up to $20 \%$) for $M = 40$ and does not necessarily converge to the true invariant as we assimilate more observations.

\subsection{A \elo with a linear invariant} 
\label{subsec:embedded_lorenz63}

In the previous numerical experiments, the filtering distribution was Gaussian. 
However, we next consider a non-Gaussian filtering problem with a linear invariant, and compare the \uncons/\cons SMFs with the \enkf. 
Recall that \Cref{algo:nonlinear_constrained} provides a pseudo-code for the \cons \smf. 
We again refer readers to \citep{spantini2022coupling} for further details on SMFs. 

Consider the three-dimensional \lo for atmospheric convection \citep{lorenz1963deterministic}: 
\begin{equation}\label{eqn:lorenz63}
\frac{\d \widetilde{\x}}{\dt} = \widetilde{\mathfrak{F}}(\widetilde{\x}, t) = 
\begin{bmatrix}
    \sigma(\widetilde{x}_2 - \widetilde{x}_1)\\
    \widetilde{x}_1(\rho- \widetilde{x}_3) - \widetilde{x}_2\\
    \widetilde{x}_1 \widetilde{x}_2-\beta \widetilde{x}_3
\end{bmatrix},
\end{equation}
where $\widetilde{\x} = [\widetilde{x}_1, \widetilde{x}_2, \widetilde{x}_3]^\top \in \real{3}$ is the state, $\widetilde{\mathfrak{F}} \colon \real{3} \times \real{} \to \real{3}$ denotes the forward operator of \eqref{eqn:lorenz63}, and $\sigma, \beta, \rho$ are fixed parameters. 
In our simulation, we use $\sigma = 10, \beta = 8/3, \rho = 28$. 
For these parameter values, the system exhibits chaotic behavior and is governed by a strange attractor \citep{lorenz1963deterministic}. 
We introduce an embedded version of the \lo with a linear invariant, which we call the \elo. 
For this low-dimensional problem, covariance tampering or localization does not help regularize the unconstrained \enkf. 
Assuming that the linear invariant is constant over the forecast distribution $\pdf{\X_{t} \given \Y_{1:t-1}}$, Section \ref{sec:invariants_kalman} proves that in this case, the unconstrained \enkf preserves linear invariants. 
This example allows us to investigate the influence of a nonlinear analysis map and the preservation of linear invariants. 
We compare the constrained stochastic map filter (\consmf) with the unconstrained ensemble Kalman filter (\unenkf) and stochastic map filter (\unsmf). 

To construct the \elo, we augment the state $\widetilde{\x} \in \real{3}$ of the original \lo of \eqref{eqn:lorenz63} by adding a fourth component $\widetilde{x}_4$ with zero dynamics, \ie $\d \widetilde{x}_4/\dt = 0$, and performing a random rotation of the embedded coordinates. 
This random rotation allows us to create a dynamical system with a non-local linear invariant, \ie the invariant depends on all the state variables. 
We denote the augmented state by $\widetilde{\x}_{\text{aug}} \in \real{4}$ and the augmented forward operator by $\widetilde{\mathfrak{F}}_{\text{aug}} \colon \real{4} \times \real{} \to \real{4}, (\widetilde{\x}_{\text{aug}}, t) \mapsto [\widetilde{\mathfrak{F}}(\widetilde{\x}, t), 0]^\top$. Observe that $\widetilde{\x} \mapsto \widetilde{\U}_\perp^\top \widetilde{\x}$ is a linear invariant for the augmented state $\widetilde{\x}$ with $\widetilde{\U}_\perp = [0, 0, 0, 1]^\top$. 
We then apply a random orthogonal matrix $\Q \in \real{4 \times 4}$ to $\widetilde{\x}_{\text{aug}}$ to define the new state $\x = \Q \widetilde{\x}_{\text{aug}}$. 
Thus, the state $\x$ is governed by the ODE system
\begin{equation}
\label{eqn:embedded_lorenz63}
    \frac{\d \x}{\dt} = \frac{\d \Q \widetilde{\x}_{\text{aug}}}{\dt} = \Q \widetilde{\mathfrak{F}}_{\text{aug}}(\Q^{-1} \x, t) = \Q \widetilde{\mathfrak{F}}_{\text{aug}}(\Q^\top \x, t),
\end{equation}
where the last equality is due to the orthonormality of $\Q$. One can verify that the linear invariant $\x \mapsto\Uperp^\top \x$ with $\Uperp = \Q \widetilde{\U}_\perp \in \real{4}$ is preserved by \eqref{eqn:embedded_lorenz63}:
\begin{equation}
    \frac{\d \Uperp^\top  \x}{\dt} 
    		=  \widetilde{\U}_\perp^\top \Q^\top  \frac{\d \x}{\dt}
    		= \widetilde{\U}_\perp^\top \widetilde{\mathfrak{F}}_{\text{aug}}(\Q^\top \x, t) 
    		= \langle [0, 0, 0, 1] , [\widetilde{\mathfrak{F}}(\Q^\top \x , t), 0] \rangle
    		= 0,
\end{equation}
where $\langle \cdot, \cdot \rangle$ denotes the Euclidian scalar product in $\real{n}$. 
The random orthogonal rotation $\Q \in \real{4 \times 4}$ is constructed by extracting the $Q$ factor of the QR factorization of a $4 \times 4 $ matrix whose entries are drawn from the standard Gaussian distribution. To solve \eqref{eqn:embedded_lorenz63}, we use a fourth-order RK time integration scheme \citep{rackauckas2017differentialequations}. 
The time step size between two assimilation cycles is $\dtobs = 5 \cdot 10^{-2}$. 
We observe every component of the state, i.e., $d = n = 4$, corrupted by additive Gaussian observation noise with zero mean and covariance $\sigma_{\Noiseobs}^2 \id{d}$ with $\sigma_{\Noiseobs} = 10^{-2}$. 
The initial distribution $\pdf{\X_0}$ is standard Gaussian. 
The true invariant $\C$ is set to $1$. 

If the observations are conditionally independent, \ie the likelihood factorizes as $\pdf{\Y \given \X} = \prod_{j=1}^d \pdf{\mathsf{Y}_j \given \X}$, it is equivalent to assimilate a $d$-dimensional observation $\ystar \in \real{d}$ at once, or to assimilate the observation components $y^\star_j \in \real{}, \; j = 1,\ldots, d,$ in $d$ recursive updates. 
We provide a justification in \ref{apx:recursive_assim}. 
Thus, the analysis step can be equivalently performed by computing $d$ analysis maps with $n+1$ inputs---each map associated with one observation, or computing a single map with $d + n$ inputs.  In practice, the number of samples to estimate the analysis map is small compared to the dimensions of the states and observations. Thus, to reduce the variance of the resulting analysis map $\tmap$, it is favorable to estimate $d$ transport maps $\smap^{\Xup}$ of $n+1$ inputs rather than a single map $\smap^{\Xup}$ of $d+n$ inputs. See \citep{spantini2022coupling, houtekamer2001sequential} for further discussions. 
Following \citep{spantini2022coupling}, we apply the \uncons/\cons stochastic map filters sequentially. 
To allow for a fair comparison, the \uncons \enkf is also applied sequentially in this example. 

\begin{figure}[tb]
    \centering
    \includegraphics[width = 0.475\linewidth]{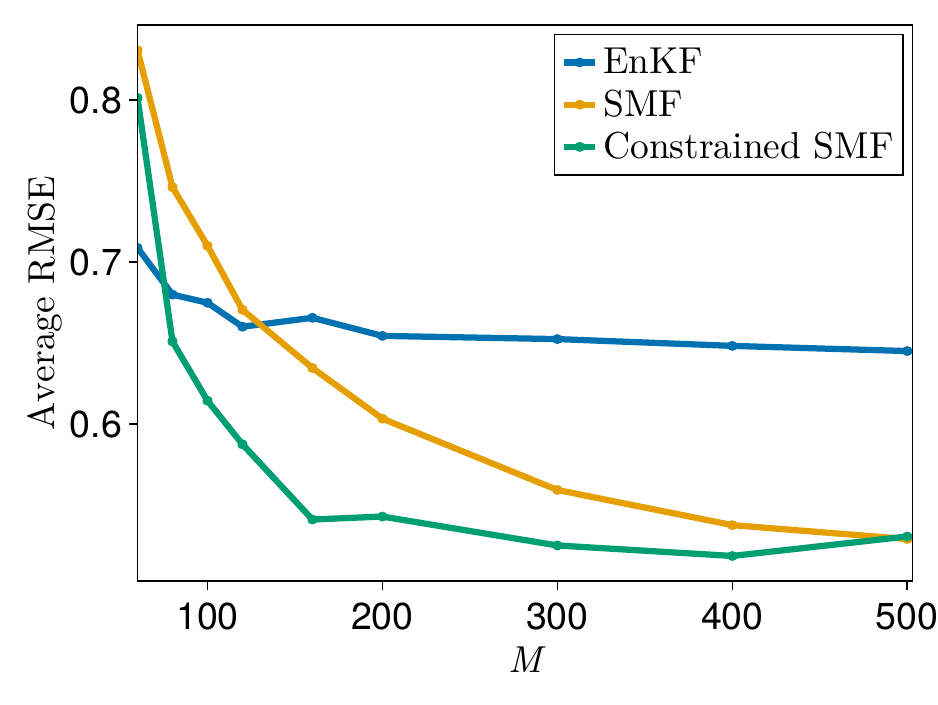}
    ~
    \includegraphics[width = 0.475\linewidth]{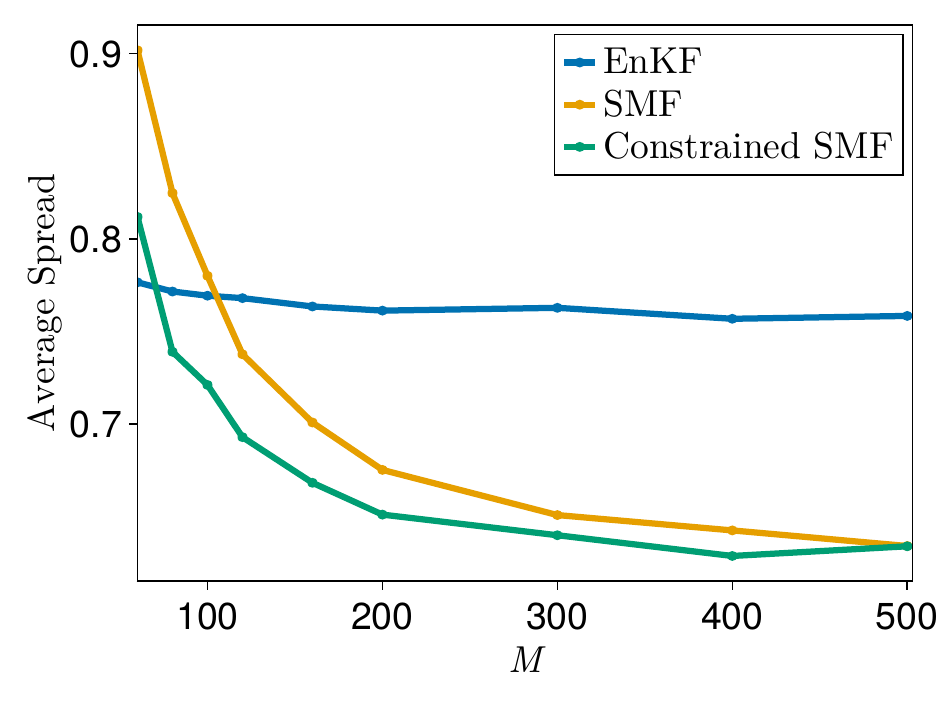}
    \caption{
        Left: Time-averaged RMSE of the \elo problem using the \unenkf (blue), the \unsmf (yellow), and the \consmf (green) for varying ensemble sizes $M$. 
        Right: Median of the spread for varying ensemble sizes $M$. 
        All filters use optimally tuned inflation. 
    }
    \label{fig:embedded_lorenz63_metrics}
\end{figure}

Figure \ref{fig:embedded_lorenz63_metrics} reports the time-averaged RMSE and the spread of the \elo using the \unenkf (blue), the \unsmf (yellow), and the \consmf (green) for varying ensemble sizes $M$. 
All filters have optimally tuned inflation. 
For small ensemble size $M < 100$, the unconstrained \enkf has a lower RMSE and spread than the stochastic map filters. For larger ensemble sizes $M > 100$, we observe consistent improvements in RMSE and spread with the nonlinear filters.  For large ensemble size $M \approx 500$, the \unsmf and \consmf perform similarly in terms of RMSE and spread, corresponding to a reduction of the RMSE by $18\%$ and the spread by $16\%$ with respect to the \enkf. These results echo the conclusions of \citep{spantini2022coupling} on the bias-variance tradeoff of nonlinear filters with limited samples. 
Interestingly, the \consmf always performs better than the \unsmf. For $M > 80$, the \consmf achieves the lowest RMSE and spread. 
For $M \in [60, 200]$, we observe that the RMSE of the \consmf rapidly decreases from $8 \cdot 10^{-1}$ for $M = 60$ to $5.3 \cdot 10^{-1}$ for $M = 200$. 
In contrast, the RMSE and spread of \unsmf exhibit slower decay with increasing ensemble size. 
The RMSE of the \consmf plateaus at about $5.3 \cdot 10^{-1}$ for $M > 160$, while the \uncons \smf requires $M = 500$ samples to achieve a similar performance. 

\begin{figure}[tb]
    \centering
    \includegraphics[width = 0.6\linewidth]{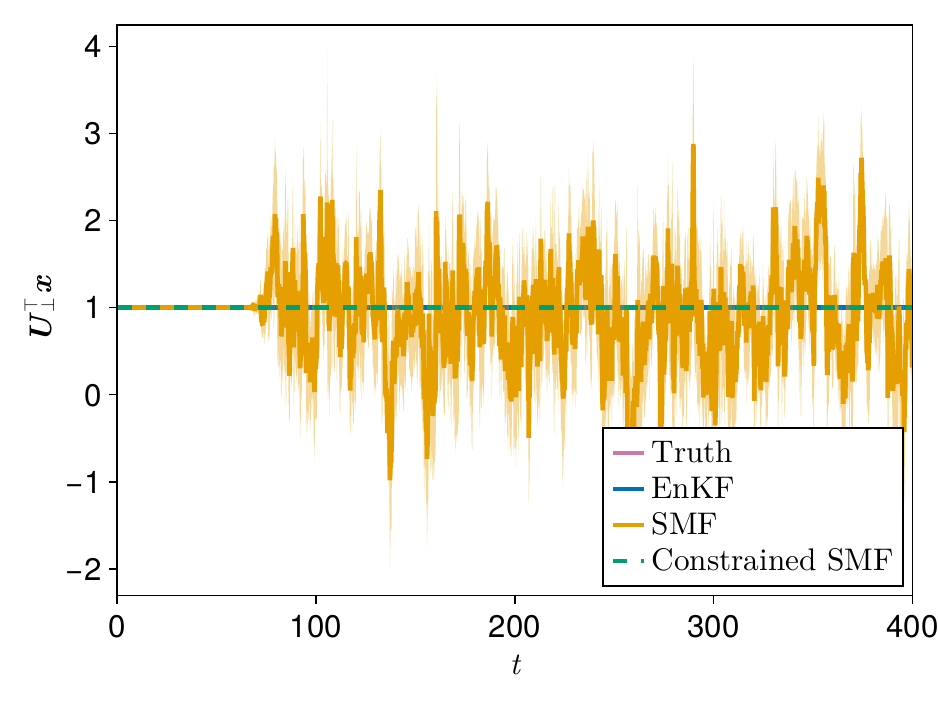}
    \vskip -0.5cm
    \caption{
        Linear invariant $\Uperp^\top \x_t$ for the true state process (pink) and the posterior mean of the \unenkf (blue), the unconstrained \smf (yellow), and the constrained \smf (dashed green) for ensemble size $M = 160$. 
        Fainted areas show the $10\%$ and $90\%$ quantiles of the posterior estimate of the invariant. 
        All filters use optimally tuned inflation. 
    }
    \label{fig:embedded_lorenz63_invariant}
\end{figure}

Furthermore, we observe from Figure \ref{fig:embedded_lorenz63_invariant} that the \unenkf and the \consmf conserve the linear invariant. 
However, the \unsmf shows significant variations in the linear invariant estimate (up to $200\%$) for $M = 160$. 
These results suggest that combining preservation of linear invariants with nonlinear prior-to-posterior transformations is beneficial for non-Gaussian filtering problems.

\begin{remark}
	In all experiments, the inflation and localization parameters are tuned to minimize the RMSE, following a common practice when the primary objective is to assess estimation accuracy. 
	We emphasize, however, that the ensemble spread is an equally important diagnostic quantity, as it reflects the filter’s representation of uncertainty and should ideally be consistent with the RMSE. 
	In the present study, our goal is to compare the relative performance of constrained and unconstrained methods under a common and controlled tuning strategy. 
	For this reason, we apply the same RMSE-based tuning procedure across all methods. 
	As observed in \cref{fig:toy_problem,fig:linadvection_metrics,fig:embedded_lorenz63_metrics}, the spread of the constrained and unconstrained filters remains of comparable magnitude, indicating that the improvements in RMSE achieved by the constrained approaches are not accompanied by a degradation of ensemble dispersion. 
	A more comprehensive calibration strategy that jointly accounts for RMSE and spread would be of clear practical interest but is beyond the scope of this work and left for future investigation.
\end{remark}
\section{Conclusion}
\label{sec:conclusion}

We have introduced \linpam s, a class of analysis maps that preserve linear invariants $\invar(\x) = \Uperp^\top \x$ of the forecast distribution in (potentially non-Gaussian) filtering problems. 
In the special case of jointly Gaussian observations and states, we recover a \cons formulation of the Kalman filter analysis map, where the update is projected onto the orthogonal complement of the columns of $\Uperp$. 
Furthermore, we have clarified existing results on preserving linear invariants for the vanilla (\uncons) Kalman filter and ensemble Kalman filter. 
In particular, we have shown that regularization techniques for the ensemble Kalman filter, such as covariance inflation, localization, or tapering, can violate linear invariants. 
We have also demonstrated how to combine these regularization techniques for the ensemble Kalman filter while preserving the linear invariants. 

The tools developed in this study are not limited to the filtering setting, as the analysis step solves a static Bayesian inverse problem \citep{leprovost2021low,leprovost2022lowenkf}. 
In fact, the techniques presented here are readily applicable to other ensemble-based methods used to solve static inverse problems while preserving linear invariants \citep{iglesias2013ensemble,zhang2020regularized, garbuno2020interacting}. 

Future works will construct analysis maps and ensemble filters for important non-linear invariants, such as Hamiltonians in mechanics \citep{del2018introduction} or energy and entropy in hyperbolic conservation laws \citep{leveque1992numerical,glaubitz2020shock}. 
For such non-linear invariants, a projection-based approach, in which one first performs an unconstrained analysis update and then projects onto the desired invariant-preserving manifold, may be computationally more practical than performing the analysis step directly on the constraint manifold.
However, the statistical interpretation of such an approach (which distribution does the ensemble approximate?) may be less clear for projection-based methods.

Finally, future work should not only address how invariants can be preserved in the analysis step of data assimilation methods, but also the equally important question of when such preservation is appropriate. 
In many settings, one may argue that quantities such as total mass or energy should be conserved for each ensemble member—either strictly in isolated systems or up to boundary fluxes in non-isolated systems—provided that reliable prior estimates of these quantities are available. 
However, when these initial estimates are uncertain or biased, enforcing strict conservation may no longer be justified. 
In such cases, it may be more appropriate to first update the estimate of the invariant itself and subsequently propagate this correction to the full state.
These considerations point to a broader, largely unexplored design space for DA algorithms in which the treatment of invariants is adaptive and informed by uncertainty. 
In this context, smoothing frameworks offer a particularly appealing perspective: 
Rather than enforcing a fixed numerical value, one may require invariants to remain consistent over time, thereby allowing their values to be inferred from data.
We believe that revisiting the role of invariants in DA—from both theoretical and algorithmic perspectives—opens a rich avenue for future research. 
By bridging physical structure, statistical inference, and computational methodology, this line of work has the potential to significantly enhance the reliability and interpretability of data-driven models. 
We hope that the ideas presented here will stimulate further developments in this direction and inspire the community to explore these opportunities.

\section*{Data Accessibility.} 
The code to reproduce the computational results is available at \href{https://github.com/mleprovost/Paper-Linear-Invariants-Ensemble-Filters}{\textcolor{magenta}{https://github.com/mleprovost/Paper-Linear-Invariants-Ensemble-Filters}}.

\section*{Funding.} 
MLP and YM acknowledge support of the National Science Foundation under Grant PHY-2028125. 
JG and YM acknowledge support of the US DOD (ONR MURI) under Grant N00014-20-1-2595.
JG acknowledges support by the Swedish Research Council (VR) Starting Grant \#2025-05370, the Zenith Career Development Grant \#26.07, and the National Academic Infrastructure for Supercomputing in Sweden (NAISS) grants \#2025/22-1599 and \#2024/22-1207.

\section*{Acknowledgments}
The authors thank Ricardo Baptista, Jeff Eldredge, Thomas Izgin, Matthew Levine, and Daniel Sharp for insightful discussions and constructive feedback.

\section*{Dedication.} 
We (JG and YM) dedicate this manuscript to the memory of our friend and colleague Mathieu Le Provost. Mathieu initiated this line of research and brought us together to work on it, but passed away before it could be completed. We hope to honor his efforts and ideas by finishing this manuscript on his behalf.

\appendix 
\section{Exact Bayesian updates preserve invariants}
 \label{apx:invariants}

\begin{theorem}
    Consider a prior $\pdf{\X}$, a likelihood $\pdf{\Y \given \X}$, and an invariant $\invar \colon \real{n} \to \real{r}$. 
    Let us assume that the invariant is constant over the prior distribution $\pdf{\X}$, \ie $\invar(\x) =  \C \in \real{r}$ for any realization $\x$ of $\X$. Then the invariant $\invar$ is also constant over the posterior distribution $\pdf{\X \given \Y}$. 
\end{theorem}

\begin{proof}
	Recall that the support of a distribution $\pdf{}$ is defined as $\supp{\pdf{}} = \{ \, \x \in \real{n} \mid \pdf{}(\x) > 0 \, \}$.
	If the invariant $\invar$ is constant over the prior distribution $\pdf{\X}$, then 
	\begin{equation}\label{eqn:support_prior}
    		\supp{\pdf{\X}} \subseteq \left\{ \, \x \in \real{n} \mid \invar(\x) = \C \, \right\}.
	\end{equation}
	Bayes' rule implies that the support of the posterior $\pdf{\X \given \Y} = \pdf{\Y \given \X} \cdot \pdf{\X}  /\pdf{\Y}$ is included in that of the prior: 
	\begin{equation}\label{eqn:support_prior_posterior}
    		\supp{\pdf{\X \given \Y}} \subseteq \supp{\pdf{\X}}
	\end{equation} 
	Indeed, a multiplication by the non-negative quantity $\pdf{\Y \given \X}(\y \given \x)/\pdf{\Y}(\y) \geq 0$ cannot increase the posterior support $\supp{\pdf{\X \given \Y}}$ beyond $\supp{\pdf{\X}}$. 
	Combining \eqref{eqn:support_prior} and \eqref{eqn:support_prior_posterior}, we conclude that 
	\begin{equation}
    		\supp{\pdf{\X \given \Y}} \subseteq \left\{ \, \x \in \real{n} \mid \invar(\x) = \C \, \right\},
	\end{equation}
	which yields the assertion.
\end{proof}
\section{Algorithm for the analysis step of the \cons ensemble Kalman filter\label{apx:linear_constrained}}

Algorithm \ref{algo:linear_constrained} presents pseudo-code for the analysis step of the \cons ensemble Kalman filter (\consenkf). 
It transforms a set of forecast samples into filtering samples by assimilating the observation $\ystar$ while preserving linear invariants $\Uperp^\top \x = \C \in \real{r}$ of the samples.

\begin{algorithm}
\caption{\texttt{consEnKF}$(\ystar, \Obs, \pdf{\Noiseobs}, \Uperp, \{\x^i\})$ assimilates $\ystar$ into $\{\x^i\}_{i=1}^M$ while preserving linear invariants}
\label{algo:linear_constrained}
\begin{algorithmic}[1]
	\State{ 
		\textbf{Input:} $\ystar \in \mathbb{R}^d$, linear observation operator $\Obs \in \mathbb{R}^{d \times n}$, observation noise distribution $\pdf{\Noiseobs} = \mathcal{N}(\mathbf{0}, \cov{\Noiseobs})$, sub-unitary matrix $\Uperp \in \mathbb{R}^{r \times n}$ for the linear invariants $\Uperp^\top \x = \C \in \mathbb{R}^r$, $M$ samples $\{\x^i\}$ from $\pdf{\X}$
	}
	\State{
		\textbf{Output:} Samples $\{\x_a^i\}_{i=1}^M$ from $\pdf{\X \mid \Y = \ystar}$
	}
	\State{
		Generate observation noise samples $\{ \noiseobs^i \}_{i=1}^M$ by drawing from $\mathcal{N}(\mathbf{0}, \cov{\Noiseobs})$
	}
	\State{
		Form perturbation matrices $\mathbf{A}_{\X} \in \mathbb{R}^{n \times M}$ and $\mathbf{A}_{\Noiseobs} \in \mathbb{R}^{d \times M}$ for state and noise: 
		$$
			\mathbf{A}_{\X}[:, i] \gets \frac{1}{\sqrt{M - 1}}(\x^i - \bar{\x}), \quad 
			\mathbf{A}_{\Noiseobs}[:, i] \gets \frac{1}{\sqrt{M - 1}}(\noiseobs^i - \bar{\noiseobs}), \quad 
			i= 1,\dots,M
		$$
	}
	\State{
		Apply Kalman gain based on representers: Solve 
		$$ 
			\left( \Obs \mathbf{A}_{\X} \mathbf{A}_{\X}^\top \Obs^\top + \mathbf{A}_{\Noiseobs} \mathbf{A}_{\Noiseobs}^\top \right) \mathbf{b}^i = \Obs \x^i + \noiseobs^i - \ystar
		$$	
	}
	\State{
		Build the posterior samples: $\x_a^i \gets \x^i - (\mathbf{I} - \Uperp \Uperp^\top) \mathbf{A}_{\X} (\Obs \mathbf{A}_{\X})^\top \mathbf{b}^i$ for $i=1,\dots,M$
	}
	\State{ 
		\Return $\{ \x_a^i \}_{i=1}^M$
	}
\end{algorithmic}
\end{algorithm}
\section{Recursive assimilation of conditionally independent observations \label{apx:recursive_assim}}

In many settings, the observations $\Y \in \real{d}$ to assimilate are conditionally independent, \ie $\mathsf{Y}_j \indep \mathsf{Y}_k \given \X$ for  $j \neq k$, \citep{spantini2022coupling, houtekamer2001sequential}. 
Then the likelihood $\pdf{\Y \given \X}$ factorizes as $\pdf{\Y \given \X} = \prod_{j=1}^d \pdf{\mathsf{Y}_j \given \X}$. 
In this case, we show that the assimilation of a $d$-dimensional observation $\ystar \in \real{d}$ can be made in $d$ recursive steps by assimilating one scalar observation $\ystar_j$ at a time in the state, \ie the posterior from assimilating $j$ components $\ystar_{1:j}$ can be used as a prior for assimilating the next scalar observation $y^\star_{j+1}$. 
From Bayes' rule, we derive a recursive relation to assimilate the observation $\mathsf{Y}_{j+1}$ in the conditional distribution $\pdf{\X \given \Y_{1:j}}$:
\begin{equation}
\label{eqn:sequential_bayes}
    \pdf{\X \given \Y_{1:j+1}} = \frac{\pdf{\Y_{1:j+1} \given \X} \pdf{\X}}{\pdf{\Y_{1:j+1}}} = \frac{\pdf{\mathsf{Y}_{j+1} \given \X, \Y_{1:j}} \pdf{\Y_{1:j} \given \X} \pdf{\X}}{\pdf{\mathsf{Y}_{j+1} \given \Y_{1:j}} \pdf{\Y_{1:j}}} = \frac{\pdf{\mathsf{Y}_{j+1} \given \X, \Y_{1:j}}} {\pdf{\mathsf{Y}_{j+1} \given \Y_{1:j}}} \pdf{\X \given \Y_{1:j}}
\end{equation}
This factorization holds for an arbitrary likelihood $\pdf{\Y \given \X}$. 
In general, the ``conditional'' likelihood $\pdf{\mathsf{Y}_{j+1} \given \X, \Y_{1:j}}$ is intractable, and one cannot easily use \eqref{eqn:sequential_bayes} to perform a recursive update of the prior without further assumptions on the likelihood $\pdf{\Y \given \X}$. 
However, if we assume that the observations are conditionally independent, \ie $\mathsf{Y}_j \indep \mathsf{Y}_k \given \X$ for  $j \neq k$, \eqref{eqn:sequential_bayes} simplifies to
\begin{equation}
\label{eqn:sequential_bayes_cond_indep}
    \pdf{\X \given \Y_{1:j+1}} = \frac{\pdf{\mathsf{Y}_{j+1} \given \X}}{\pdf{\mathsf{Y}_j}} \pdf{\X \given \Y_{1:j}},
\end{equation}
and one can assimilate the observations recursively.

\bibliographystyle{elsarticle-num-names}
\bibliography{references}

\end{document}